\documentclass[12pt]{article}

\AtBeginDocument{%
   \setlength\abovedisplayskip{0.05in}
   \setlength\belowdisplayskip{0.08in}}

\usepackage[onehalfspacing]{setspace}

\makeatletter
\renewcommand\@footnotetext[1]{%
  \insert\footins{%
    \reset@font\footnotesize
    \interlinepenalty\interfootnotelinepenalty
    \splittopskip\footnotesep
    \splitmaxdepth \dp\strutbox \floatingpenalty \@MM
    \hsize\columnwidth \@parboxrestore
    {\setstretch{1.0}\protect\@makefntext{%
      \rule{\z@}{\footnotesep}\ignorespaces#1}}}}
\makeatother

\usepackage[top=50mm, bottom=-500mm, left=50mm, right=50mm]{geometry}
\usepackage{lineno}
\usepackage{amssymb}
\usepackage{amsthm}
\usepackage{amsmath}
\usepackage{cases}
\usepackage{caption}
\usepackage{epsfig}
\usepackage{graphicx}
\usepackage{graphics}
\usepackage{float}
\usepackage{subfigure}
\usepackage{multirow}
\usepackage{xcolor}
\usepackage{fullpage}
\usepackage[normalem]{ulem}
\usepackage{makeidx}
\usepackage{xspace}
\usepackage{wrapfig}
\usepackage{lscape}
\usepackage{mathrsfs}
\usepackage{adjustbox}
\usepackage{varwidth}
\usepackage{rotating}
\usepackage[bottom]{footmisc}
\usepackage[colorlinks,citecolor=darkblue,urlcolor=darkblue,linkcolor=darkblue]{hyperref}
\usepackage{cleveref}
\makeindex
\newtheorem{theorem}{Theorem}

\newtheorem{corollary}{Corollary}

\newtheorem{lemma}{Lemma}

\newtheorem{proposition}{Proposition}
\newtheorem{definition}{Definition}

\newtheorem{obs}{Observation}

\definecolor{purple}{rgb}{0.6, 0.4, 0.8}
\definecolor{darkred}{rgb}{1, 0.1, 0.3}
\definecolor{darkblue}{rgb}{0.0, 0.0, 0.55}
\definecolor{darkgreen}{rgb}{0,0.6,0.5}
\definecolor{darkpurple}{rgb}{0.4, 0.0, 0.6}
\definecolor{darkorange}{rgb}{0.9, 0.4, 0.0}
\definecolor{forestgreen}{rgb}{0.0, 0.46, 0.37}
\definecolor{bittersweet}{rgb}{1.0, 0.44, 0.37}
\definecolor{navy}{rgb}{0.0, 0.0, 0.55}
\definecolor{brown}{rgb}{0.53, 0.18, 0.09}
\definecolor{Green}{rgb}{0.0, 0.47, 0.44}

\newcommand {\mm}[1] {\ifmmode{#1}\else{\mbox{$#1$}}\fi}

\newcommand\E{\mathbb{E}}

\newcommand\R{\mathbb{R}}

\newcommand{\co}{\operatorname{co}}
\newcommand{\supp}{\operatorname{supp}}
\newcommand{\KL}{D_{\mathrm{KL}}}

\usepackage{dsfont}

\usepackage{sectsty}
\sectionfont{\centering\Large}
\subsectionfont{\large}
\usepackage{fancyhdr}
\usepackage{natbib}

\usepackage{pst-node, pst-plot, pst-poly}
\usepackage{tikz}
\usepackage{venndiagram}
\usepackage{pstricks}
\usepackage{tikz-cd}
\usepackage{siunitx}
\usepackage{pgfplots}
\usepgfplotslibrary{fillbetween}
\usepackage{algorithm}
\usepackage{algorithmic}
\usepackage{bbm}
\usepackage{csquotes}
\usepackage{scalerel}

\usepackage[T1]{fontenc}
\usepackage{titling}
\usetikzlibrary{shapes,arrows,trees,calc}
\pgfplotsset{compat=1.12}

\allowdisplaybreaks

\begin{document}

\setlength{\droptitle}{-1.1in}
\title{\small\bf \mbox{\MakeUppercase{Constraint Preferences: Inattention and Aggregation}}\thanks{I sincerely thank Drew Fudenberg and Stephen Morris for their guidance and support during this project.}
\vspace{-0.25in}
}
\author{\textsc{Florian Mudekereza}\thanks{Department of Economics, MIT, \href{mailto:florianm@mit.edu}{\texttt{\footnotesize florianm@mit.edu}}.}}
\date{}
\maketitle
\thispagestyle{empty}
\setcounter{page}{0}
\vspace{-0.88in}

\begin{abstract}
We study robust decision problems when individuals have maxmin preferences whose belief sets are neighborhoods around reference models, commonly known as ``constraint preferences.'' We first show that a more disciplined form of rationally inattentive behavior is equivalent to the behavior implied by a subclass of constraint preferences. We then introduce an aggregation principle that requires collective beliefs to satisfy every individual's constraint. This requirement links collective beliefs to the information-processing technologies that generate individual constraints. Applications reveal how these technologies determine asset prices, when prediction-market prices become self-confirming, and how much dynamic mechanisms can reduce information rents.
\end{abstract}

\section{Introduction}\label{sec:introduction}

Economic decisions under uncertainty are often guided by a reference model, such as a random walk for forecasting returns or the Black-Scholes model for valuing options.  Viewing this model as an approximation of the true data-generating process, a decision maker may treat its deviations as plausible only within a given neighborhood and evaluate decisions under the least favorable model. This is the logic of \textit{constraint preferences} in robust decision theory \citep[see,][]{hansen01,hansen08,ball20,chamb20}. Such preferences are special cases of \citeauthor{gilboa89}'s (\citeyear{gilboa89}) maxmin preferences and provide a simple language for describing model uncertainty because they separate three objects: a reference model, the geometry used to measure deviations from it, and the size of the neighborhood. However, this raises a natural question: where do these neighborhoods come from? Our first result provides an information-processing foundation for a subclass of constraint preferences. 
\par Before making a decision, suppose an individual can update their reference model (or belief) by running a Blackwell experiment. Let the cost of running this experiment be \textit{uniformly posterior-separable}, i.e., it is the probability-weighted average of discrepancies---measured by Bregman divergences---between the reference belief and the posteriors it produces \citep{caplin_dean_leahy22}.\footnote{Uniform posterior separability is popular because it has sequential-sampling foundations \citep[e.g.,][]{ups19,HebertWoodford2021,seq23} and experimental support \citep[e.g.,][]{exp23}.} Theorem \ref{thm:strict_ri} establishes a behavioral equivalence: the minimum information cost per unit probability of producing a posterior is the Bregman divergence between that posterior and the reference belief. We refer to this minimum cost as the posterior's \emph{unit cost}; imposing a bound on this cost therefore restricts admissible beliefs to a Bregman ball around the reference belief. Thus, robust decision-making subject to this bound is equivalent to having a constraint preference defined over that Bregman ball. Bregman divergences form a broad family  of statistical distances between probability distributions that are widely used in economics, such as the relative entropy \citep[e.g.,][]{breg19,ri20,HebertWoodford2021,seq23,caplin_dean_leahy22,breg25}, and the associated balls are also popular in economics, finance, and statistics \citep[e.g.,][]{hansen01,hansen08,watson16,ball20,fin26}.

We define a behavior called \emph{strict rational inattention} to describe the information-processing constraint characterized by Theorem \ref{thm:strict_ri}. The behavioral equivalence clarifies how it differs from canonical rational inattention, which constrains the expected cost of an experiment chosen before its information is processed \citep{sims03}; see \citet{survey23} for a survey. Strict rational inattention instead constrains the unit cost of producing a posterior after information has been processed. Under uniform posterior separability, this distinction matters because making a posterior sufficiently rare can make its contribution to expected cost arbitrarily small without changing its unit cost (Observation \ref{prop:degeneracy}). Strict rational inattention therefore controls the cost of a posterior independently of how frequently that posterior is realized. The timing difference between strict and canonical rational inattention arises naturally in many economic applications, such as Bayesian privacy and persuasion. In privacy-constrained mechanisms, for example, \citet{eilat2021bayesian} study an
ex post privacy constraint that bounds the relative entropy between a prior and the resulting posterior
after every signal, and an ex ante constraint that bounds its expectation---the mutual information.\footnote{\citet{babichenko21} study a similar version of ex-post vs. ex-ante constraints in Bayesian persuasion.} Thus, these two constraints coincide, respectively, with strict and canonical rational inattention. Their behavioral difference is that the ex ante constraint is more permissive because large privacy
losses may occur on low-probability realizations as long as average privacy loss
remains small, whereas the ex post constraint limits privacy loss at every
realization.

When  decisions become collective, constraint preferences raise an aggregation problem because a social planner needs to construct collective beliefs from individual belief sets. Our second result proposes an aggregation principle that translates individual belief constraints into 
constraints on collective beliefs. While maxmin preferences are complete, their belief sets also identify comparisons that are robust to ambiguity in the sense of \citet{bew02}: a policy is unambiguously preferred to another when every belief entering the maxmin criterion supports that comparison. We introduce \emph{unambiguous maxmin Pareto dominance}, which requires the planner's maxmin preference to respect every comparison that is unanimously unambiguous for individuals. Under a standard diversity condition on tastes, Theorem \ref{thm:agg} characterizes this principle: the social utility must be a linear combination of individual utilities, while every social belief must satisfy the belief constraints of all relevant individuals.

The belief restriction in Theorem \ref{thm:agg} is an intersection: the collective belief set must contain only models that are shared by all relevant individuals. This has a natural interpretation as social agreement under imprecise beliefs. Such intersections are used as ``domains of consensus'' in public policy \citep{manski95}, as ``shared beliefs'' in organizations \citep{van10} and markets \citep{gil00,RigottiShannon2005}, and more broadly in opinion aggregation \citep[e.g.,][]{Walley91,inter02,robust16,niel18,inter18,hill16,hill23}. Our aggregation principle differs from \citet{robust16}, who take incomplete unambiguous preferences as the individual and social preferences, and impose Pareto directly on those relations. Here, both individual and social preferences are maxmin, and Pareto requires only the planner's ranking to respect unanimous unambiguous comparisons. Thus, our principle is less demanding than \citeauthor{robust16}'s (\citeyear{robust16}) unambiguous Pareto dominance while still delivering the same aggregation rule (Observation \ref{prop:pareto_comparison}). 

The first two results separate the operation used to aggregate beliefs from the information technology that determines the geometry of individual belief sets. We combine these insights to show that the Bregman geometry is economically relevant because it is rich enough to nest many of the popular belief-aggregation rules used in economics and finance. The key is that Bregman balls embed an extra degree of freedom: the \textit{order} of their arguments matters because Bregman divergences are typically asymmetric. The two orders are referred to as \textit{primal} and \textit{dual}, and are derived from different information technologies: uniformly posterior-separable costs deliver dual Bregman balls, whereas scoring-rule regrets deliver primal Bregman balls. These orders pin down different belief-aggregation rules. On one hand, primal Bregman balls can select linear pooling of  beliefs, so they recover \citeauthor{gil04}'s (\citeyear{gil04}) linear-linear welfare criterion, and the convex-hull belief set in \citeauthor{util16}'s (\citeyear{util16}) maxmin welfare criterion and \citeauthor{simsek14}'s (\citeyear{simsek14}) belief-neutral welfare criterion. On the other hand, dual Bregman balls can select geometric pooling of beliefs, so they recover \citeauthor{diet21}'s (\citeyear{diet21}) linear-geometric welfare criterion; they can also select the power pooling rules used in welfare economics for inequality measurement \citep{maa86} and in finance for asset pricing \citep{asset21}; we then show that they can recover the multiplicative pooling rule axiomatized by \citet{diet10}.  It is also possible to include all these familiar pooling rules in the same intersection, thereby allowing the planner to have access to all of them at once. This property of our aggregation principle is consistent with a central theme in epistemic social choice, which suggests that beliefs should be aggregated according to the information technology used to form them \citep[e.g.,][]{fleu18,pivato19,update20,pivato22,pivato24}.

We then explore three applications that illustrate different aspects of our framework. The financial-market application focuses on the aggregation principle itself. Building on \citet{gil00}, shared beliefs in the intersection of individual belief sets support full insurance and are capitalized into equilibrium prices. The Bregman geometry then determines which shared belief is selected when the intersection is a singleton. This distinction is economically consequential: we show that using an aggregation rule associated with the wrong information technology can lead a planner to approve trades that generate aggregate financial losses. The other two applications focus more on the geometry of individual constraint sets. The prediction-market application builds on \citet{agg24} and shows how the Bregman geometry governs belief updating after private signals. It identifies the information capacity required to rationalize an inherited price and therefore determines when that price becomes self-confirming (Theorem \ref{prop:no_learning_frontier}). The mechanism-design application builds on \citet{strack25} to study a different implication of the same geometry: least-favorable beliefs determine how much dynamic screening can relax static incentive constraints. 

\par\noindent\textit{--- Related literature}: This paper contributes to the rational-inattention and robust-control literature, pioneered by \citet{sims03} and \citet{hansen01}, respectively, by introducing  strict rational inattention. This helps establish an equivalence between rationally inattentive behavior and  robust-control behavior implied by constraint maxmin preferences. 
\par We also contribute to epistemic social choice theory, which is the study of collective decision problems in which disagreement concerns not only values or tastes but also facts, probabilities, and information.  In environments where there are objective probabilistic facts, the correct way to aggregate individuals' beliefs about these facts depends on specific assumptions on how they form these beliefs. In a survey, \citet{pivato19} argues that in such environments,  epistemic social choice should start from a model of belief formation and then derive the appropriate aggregation rule from that model. In a related survey, \citet[][p. 34]{fleu18} notes: ``The aggregation of beliefs should be completely separated from, and be performed prior to, the aggregation of preferences.'' These methodological points are reflected in \citet{simsek14}: they assume individual belief formation is distorted by psychological biases and propose their convex-hull aggregation rule to address this problem. Their approach illustrates that assumptions about belief formation naturally determine which belief-aggregation rule is appropriate. Our framework follows these perspectives by deriving collective beliefs from individual information technologies, separate from preferences. 

\par We relate to the broader literature on welfare aggregation under heterogeneous beliefs. Starting with \citet{hars55}, the classical approach derives utilitarian aggregation from the standard Pareto principle when uncertainty is objective. However, \citet{mongin95} shows that this approach is incompatible with subjective heterogeneous beliefs. The literature following \citet{gil04} therefore weakens Pareto-style axioms to accommodate belief heterogeneity while disciplining collective utility and collective beliefs jointly \citep[e.g.,][]{util16,robust16,qu17,hayashi_lombardi19,brandl21,billot21,diet21}. A growing strand of this literature deliberately separates the ethical aggregation of utilities from the epistemic aggregation of beliefs. \citet{update20} distinguish objective from subjective uncertainty to limit where common beliefs are imposed. \citet{pivato22} uses accumulating evidence to obtain utilitarian social welfare without determining collective beliefs. \citet{pivato24} use almost-objective uncertainty to decouple utility aggregation from belief aggregation, which is then treated as a distinct epistemic problem. We complement this literature by providing a rational-inattention foundation for such epistemic problems, which explains why linear, geometric, power, and multiplicative pooling arise from different models of belief formation. 

\par Another connection is to the literature that aggregates maxmin belief sets through a convex-mixture rule. \citet{gil11} characterize a maxmin welfare criterion whose beliefs are weighted averages of individuals' beliefs; \citet{qu17} characterizes separate linear aggregation of beliefs and tastes under maxmin preferences. Although these results differ in their axioms and domains, they share a convex-mixture logic: collective beliefs may retain models admitted by only some individuals, or mixtures of those models. In contrast, our intersection rule is deliberately less permissive because it requires every collective belief to satisfy every relevant individual's constraint. Thus, when individual belief sets are generated by information-processing constraints, changes in their geometry translate directly into changes in the admissible collective beliefs, rather than being absorbed into a broad convex envelope.
\par\noindent\textit{--- Outline}. The  paper proceeds as follows. Section \ref{sec:framework} describes constraint preferences, and Section \ref{sec:strict_ri} develops their information-processing foundation. Section \ref{sec:aggregation} introduces the aggregation principle. Section \ref{sec:bregman} studies how information technologies shape collective beliefs. Section \ref{sec:app} presents applications, and Section \ref{sec:conclusion} concludes. All proofs are in Appendix \ref{app:proofs}.

\section{Framework}\label{sec:framework}

We consider a finite set of states $S$ with $|S|\geq2$, a convex set of outcomes $X$, and the set $F:=X^S$ of acts. Let $\Delta:=\Delta(S)$ denote the set of all probability distributions over $S$ and  $\Delta^\circ:=\big\{p\in\Delta:p(s)>0\hspace{0.05in} \forall s\in S\big\}$. An outcome $x\in X$ is identified with the constant act yielding outcome $x$ no matter which state occurs. We first focus on a single decision maker. A preference over acts is a binary relation $\succsim$ on $F$. We write $f\succsim g$ when act $f$ is weakly preferred to $g$, while strict preference and indifference are denoted by $\succ$ and $\sim$, respectively. 

\subsection{Constraint preferences}\label{sec:constraint_pref}

Our primitive is a maxmin preference whose belief set is a fixed neighborhood around a reference belief. The geometry of this neighborhood is determined by a divergence measure 
$D:\Delta\times\Delta
        \rightarrow
        \R_+\cup\{+\infty\}$
such that $D(p\lVert p)=0$ for every $p\in\Delta$. We focus on divergences where, for every $p\in\Delta$, the map
$q\mapsto D(q\lVert p)$
is convex and lower semicontinuous on $\Delta$. 

For a reference belief $p\in\Delta$ and radius $\eta\geq0$, define the corresponding \emph{constraint ball}
\begin{equation*}
        \mathcal B^D_\eta(p)
        :=
        \{q\in\Delta:D(q\lVert p)\leq\eta\},
\end{equation*}
which is nonempty, convex, and compact. A constraint preference is a maxmin preference whose belief set is $\mathcal B^D_\eta(p)$; the divergence $D$ determines which distortion of $p$ are regarded as comparable in size, while $\eta$ determines how large a distortion can be admitted. 
\begin{definition}[Constraint preference]\label{def:constraint}\normalfont
A binary relation $\succsim$ on $F$ is a \emph{constraint preference} if there exist a nonconstant affine utility function $u:X\to\R$, a reference belief $p\in\Delta$, a divergence $D$, and a radius $\eta\geq0$ such that
$V(f)
        :=
        \min_{q\in\mathcal B^D_\eta(p)}
        \E_q[u(f)]$
represents $\succsim$.
\end{definition}

This formulation nests some popular constraint rules. If
$D(q\lVert p)=\|q-p\|_2$,
then $\mathcal B^D_\eta(p)$ is a Euclidean ball around $p$, as used in \citet{ball20}.\footnote{More generally, \citet{ball20} axiomatizes constraint preferences where the radius can be endogenous.} If
$D(q\lVert p)=\KL(q\lVert p)$ is the relative entropy,
then $\mathcal B^{\KL}_\eta(p)$ is the entropy neighborhood used in the canonical constraint preferences of \citet{hansen01}.\footnote{Constraint preferences with entropy balls are closely related to the multiplier preferences axiomatized by \citet{strz11}, which are special cases of \citeauthor{mmr06}'s (\citeyear{mmr06}) variational preferences.} 
If
$D(q\lVert p)=D_{\phi}(q\lVert p)$ is a $\phi$-divergence, then $\mathcal B^{D_{\phi}}_\eta(p)$ is the $\phi$-divergence ball used in \citet{chamb20}. If
$D(q\lVert p)=\sqrt{D_{\alpha}(q\lVert p)}$ is the square root of an $\alpha$-divergence, then $\mathcal B^{D_{\alpha}}_\eta(p)$ is the $\alpha$-divergence ball of \citet{kitamura1}.  We will often write a constraint preference as the triplet $(u,p,P)$, with
$P=\mathcal B^D_\eta(p)$.

\subsection{Dual Bregman constraint preferences}\label{sec:dual_bregman}

A key subclass of constraint preferences is obtained when the divergence $D$ in Definition \ref{def:constraint} is a \textit{Bregman} divergence \citep{breg67}. For a convex differentiable generator $G$, let
\begin{equation}\label{eq:bregdiv}
        D_G(q\lVert p)
        :=
        G(q)-G(p)-\langle\nabla G(p),q-p\rangle
\end{equation}
be the Bregman divergence, where
$D_G(q\lVert p)\geq0$ for all $p,q\in\Delta$ and $D_G(p\lVert p)=0$.
Then,
\begin{equation}\label{eq:dualball}
        \hat B^G_\eta(p)
        :=
        \big\{q\in\Delta:D_G(q\lVert p)\leq\eta\big\}
\end{equation}
is called a \emph{dual} Bregman ball in information geometry \citep[see,][]{topo18,ball21}. Under the regularity conditions in Appendix \ref{app:conditions}, $D_G(\cdot\lVert p)$ is an admissible divergence in the sense of Section \ref{sec:constraint_pref}, so
$\hat B^G_\eta(p)
        =
        \mathcal B^{D_G}_\eta(p).$ For example, when $G(q)=\sum_{s\in S}q(s)\log q(s)$, 
we obtain the relative entropy $\KL(q\lVert p):=\sum_{s\in S}q(s)\log\frac{q(s)}{p(s)}$; when $G(q)=\lVert q\lVert_2^2$, we obtain the squared Euclidean distance $D_G(q\lVert p)=\lVert q-p\lVert_2^2$.

\section{Rational Inattention}\label{sec:strict_ri}

We now give dual-Bregman constraint preferences an information-processing foundation. Let $Y$ be a finite set of signals. A signal structure is a collection $\sigma=\{\sigma_s\in\Delta(Y)\}_{s\in S}$, where $\sigma_s(y)$ is the probability of signal $y$ in state $s$; equivalently, $\sigma$ is a finite Blackwell experiment \citep{blackwell51}. Given a full-support reference belief $p\in\Delta^\circ$, the unconditional probability of signal $y$ is
$\pi_y(\sigma,p):=\sum_{s\in S}p(s)\sigma_s(y)$,
and whenever this probability is positive, the posterior is
 $q_y(\sigma,p)(s):=p(s)\sigma_s(y)/\pi_y(\sigma,p)$. This is the standard costly-information-acquisition problem: before acting, an individual can run a Blackwell experiment---an information structure that produces a signal---and then updates the reference belief using Bayes rule.

An information cost function assigns to each finite signal structure $(Y,\sigma)$ a cost, denoted $C(\sigma,p;Y)\in\R_+$. Motivated by the theoretical results of \citet{HebertWoodford2021,seq23}, \citet{caplin_dean_leahy22}, and \citet{ups25}; and the experimental evidence of \citet{exp23}, we focus on cost functions that are \textit{uniformly posterior separable}. 

\begin{definition}[Uniform posterior separability]\label{def:ups}\normalfont
There exists a convex differentiable generator $G$ such that, for every $p\in\Delta^\circ$ and every finite signal structure $(Y,\sigma)$,
\begin{equation}\label{eq:UPS}
        C(\sigma,p;Y)
        =
        \sum_{y\in Y}\pi_y(\sigma,p)
        D_G(q_y(\sigma,p)\lVert p),
\end{equation}
where $D_G(\cdot\lVert \cdot)$ denotes the Bregman divergence in \eqref{eq:bregdiv}.
\end{definition}

Uniform posterior separability was introduced by \citet{caplin_dean_leahy22}. It is justified in \citet{HebertWoodford2021,seq23} by a sequential-sampling foundation: optimal evidence accumulation can induce uniformly posterior-separable costs, so \eqref{eq:UPS} can be viewed as a static reduced-form representation.
\eqref{eq:UPS} says that the total cost of an experiment is the average (over realized posteriors) of a local ``difficulty'' measure $D_G$, where
the function $G$ summarizes the local discriminability properties of the underlying sampling process; the relevant belief domain of $G$ is standard but technical, so it is stated in Appendix \ref{app:conditions}. For example, given any posterior $q$, if $D_G(q\lVert p)=\KL(q\lVert p)$, then \eqref{eq:UPS} becomes \citeauthor{sims03}'s (\citeyear{sims03}) mutual information.

\subsection{Strict rational inattention}

To connect the experiment-level cost in \eqref{eq:UPS} to the belief set of a constraint preference, we isolate the information cost associated with producing a particular posterior. This leads to the following definition of a posterior's ``unit cost.''

\begin{definition}[Unit cost]\label{def:unitcost}\normalfont
Fix $p\in\Delta^\circ$. For any posterior $q\in\Delta$, its \textit{unit cost} is defined as
\begin{equation}\label{eq:kappa}
\kappa(q|p):=
\inf\left\{
\frac{C(\sigma,p;Y)}{\pi_{y^\star}(\sigma,p)}
:
\begin{array}{l}
Y\text{ finite},\hspace{0.04in}y^\star\in Y,\\
\pi_{y^\star}(\sigma,p)>0,\\
q_{y^\star}(\sigma,p)=q
\end{array}
\right\}.
\end{equation}
\end{definition}

The normalization in \eqref{eq:kappa} divides the expected cost of a signal structure in \eqref{eq:UPS} by the probability of the realization that induces $q$ given the reference belief $p$. It therefore converts an ex-ante information cost into the cost per realized occurrence of a specific posterior.

\begin{definition}[Strict rational inattention]\label{def:strictRI}\normalfont
A constraint preference $(u,p,P)$ is \emph{strictly rationally inattentive} under information cost $C$ and capacity $\eta\geq0$ if
$P=\{q\in\Delta:\kappa(q|p)\leq\eta\}.$
\end{definition}

As noted in Section \ref{sec:strict_stronger}, we use ``strict'' to distinguish the unit-cost restriction $\kappa(\cdot|p)\leq\eta$ from the canonical rational-inattention restriction on expected cost $C(\cdot,p;\cdot)\leq\eta$;\footnote{See \citet[][Section 2.4.2]{survey23} for a survey on the canonical capacity constraint. For example,  \citet{cap10}, \citet{mondria13}, and \citet{cap16} impose such ``hard'' constraints in various applications such as portfolio choice and mutual funds.} a posterior is admissible under the former only when its unit cost does not exceed the capacity budget.

\begin{theorem}[Characterization]\label{thm:strict_ri}
Fix $p\in\Delta^\circ$ and suppose $C$ is uniformly posterior-separable with generator $G$. Then,
$\kappa(\cdot|p)=D_G(\cdot\lVert p)$. Thus, for every $\eta\geq0$, a constraint preference $(u,p,P)$ is strictly rationally inattentive under $(C,\eta)$ if and only if
$P=\hat B^G_\eta(p).$
\end{theorem}

Theorem \ref{thm:strict_ri} gives constraint belief sets a simple economic interpretation: the minimum cost per unit probability of producing a particular posterior $q$ from the reference belief $p$, denoted $\kappa(q|p)$, is exactly $D_G(q\lVert p)$. A Bregman ball is therefore the set of posteriors that can be produced from that reference belief without exceeding a fixed unit-cost budget. Conversely, a maxmin preference over the resulting Bregman ball can be interpreted as robust decision-making subject to a constraint on the production capacity of posterior beliefs.

\par It is known from \citet[][Theorem 3]{caplin_dean_leahy22} that \citeauthor{sims03}'s (\citeyear{sims03}) mutual information, also called Shannon cost function, is uniquely identified within the uniformly posterior-separable family by an axiom called ``invariance under compression.''\footnote{This axiom states that behavior is invariant to changes in the decision problem that leave the probabilistic structure of payoffs unchanged---meaning the nature of the state space itself is behaviorally irrelevant.} Thus, imposing that axiom in Theorem \ref{thm:strict_ri} would pin down the dual entropy ball $\hat B^{\text{\normalfont\tiny KL}}_{\eta_i}(p_i)$ as the constraint set.

\subsection{Strict rational inattention is stronger}\label{sec:strict_stronger}

The distinction between expected and unit cost follows from the constant-marginal-cost property of uniformly posterior-separable information costs. To see this, fix a signal structure $(Y,\sigma)$ and let $o\notin Y$ be a null signal realization. For $\alpha\in[0,1]$, define the \emph{dilution} $\alpha\cdot(Y,\sigma)$ as the signal structure on $Y\cup\{o\}$ given by
$\sigma_s^\alpha(y):=\alpha\sigma_s(y)$ $\forall y\in Y$ and $\sigma_s^\alpha(o):=1-\alpha.$
The diluted Blackwell experiment first randomizes between activating the original experiment with probability $\alpha$ and producing an uninformative null signal with probability $1-\alpha$.

\begin{lemma}\label{thm:dilution}
For every $(Y,\sigma)$, $p\in\Delta^\circ$, and $\alpha\in[0,1]$,
$C(\sigma^\alpha,p;Y\cup\{o\})
        =
        \alpha C(\sigma,p;Y).$
\end{lemma}

Lemma \ref{thm:dilution} is the dilution, or constant-marginal-cost, property of \citet[][Axiom 3]{pts23}---the marginal cost of increasing the success probability of an experiment is constant. This property has a sharp implication in our framework. 

\begin{obs}[Degeneracy of expected cost]\label{prop:degeneracy}
Fix $p\in\Delta^\circ$ and a target posterior $q\in\Delta$. Suppose there exists a signal structure $(Y,\sigma)$ and $y^\star\in Y$ such that $\pi_{y^\star}(\sigma,p)>0$ and $q_{y^\star}(\sigma,p)=q$. Then, for every budget level $\eta>0$, there exists a signal structure $(\widetilde Y,\widetilde\sigma)$ and $\widetilde y^\star\in\widetilde Y$ such that
$C(\widetilde\sigma,p;\widetilde Y)\leq\eta$ and $q_{\widetilde y^\star}(\widetilde\sigma,p)=q.$
Thus, the same posterior can remain attainable while total expected cost can be made arbitrarily small.
\end{obs}

Observation \ref{prop:degeneracy} isolates the distinction between experiment-level and posterior-level information constraints. In canonical rational inattention \citep[e.g.,][]{sims03}, expected information cost disciplines the amount of information acquired on average by assigning a cost to the entire experiment. It does not, however, impose a nontrivial restriction on which specific posteriors may occur: any attainable posterior can be made arbitrarily rare, which would then reduce its contribution to expected cost. In contrast, strict rational inattention rules out precisely this margin: it controls the cost per unit probability of producing the posterior. 

\begin{obs}\label{prop:unitprice}
Fix $p\in\Delta^\circ$ and $q\in\Delta$ such that $\kappa(q|p)<+\infty$. For $\varepsilon\in(0,1]$, let
$$c^\star(\varepsilon)
        :=
        \inf\big\{
        C(\sigma,p;Y):
        \exists y^\star\in Y
        \text{ with }
        q_{y^\star}(\sigma,p)=q,\hspace{0.03in}\pi_{y^\star}(\sigma,p)\geq\varepsilon
        \big\},$$
with the convention $\inf\emptyset=+\infty$. Then, the following hold:
\begin{enumerate}
\item For every $\varepsilon\in(0,1]$,
$c^\star(\varepsilon)\geq\varepsilon\kappa(q|p)$.
\item For every $\delta>0$, there exists $\bar\varepsilon\in(0,1]$ such that, for all $\varepsilon\in(0,\bar\varepsilon]$,
$c^\star(\varepsilon)\leq\varepsilon(\kappa(q|p)+\delta)$.
\end{enumerate}
\end{obs}

Observation \ref{prop:unitprice} gives the unit cost in Definition \ref{def:unitcost} an operational interpretation. The quantity $c^\star(\varepsilon)$ is the minimum expected cost of a signal structure in which some signal realization inducing $q$ occurs with probability at least $\varepsilon$. Therefore,
$c^\star(\varepsilon)
        =
        \varepsilon\kappa(q|p)+o(\varepsilon)
        $ $\text{as }\varepsilon\downarrow0.$
That is, $\kappa(q|p)$ is the marginal expected cost of making posterior $q$ occur with small probability. Appendix \ref{app:unitcost_unique} strengthens this point by showing that $\kappa$ is the unique nondegenerate posterior-cost index calibrated to the expected  cost. Observations \ref{prop:degeneracy} and \ref{prop:unitprice} therefore make the sense in which strict rational inattention is stronger precise: an expected-cost restriction is vacuous when the posterior can be made rare, while a unit-cost restriction continues to exclude every posterior whose informational cost exceeds its capacity budget.

\subsection{Primal Bregman balls}\label{sec:primal_framework}

Since Bregman divergences are generally asymmetric, reversing their arguments produces a different divergence. Define
$D_G^{\leftarrow}(q\lVert p)
        :=
        D_G(p\lVert q).$
Under the regularity conditions in Appendix \ref{app:conditions},
$D_G^{\leftarrow}(\cdot\lVert p)$
is convex and lower semicontinuous on $\Delta$ and therefore is an admissible divergence in the sense of Section \ref{sec:constraint_pref}. Its constraint ball is
$B^G_\eta(p)
        :=
        \mathcal B^{D_G^{\leftarrow}}_\eta(p)
        =
        \{q\in\Delta:D_G(p\lVert q)\leq\eta\}.$
We call this a \emph{primal} Bregman ball. Primal and dual balls coincide only for symmetric Bregman divergences, such as Euclidean distance up to scale. For example, in the entropy case, the distinction is between primal $\KL(p\lVert q)$ and dual $\KL(q\lVert p)$.\footnote{In statistics, \citet[][Section 4.1]{watson16} discuss dual and primal entropy balls.} \citet[][Figure 2]{topo18} visualize the differences between  primal and dual balls. 

The key point is that the two orders capture different economic behaviors. Theorem \ref{thm:strict_ri} shows that the dual order is generated by uniform posterior separability: the local cost of a posterior is measured from the posterior toward the reference belief. Appendix \ref{sec:micro_primal} shows that the primal order arises when the local cost is instead the regret from validating the posterior against the reference model using a proper scoring rule. Section \ref{sec:bregman} will compare how the aggregation of these two types of constraint balls restrict collective beliefs.
\section{Aggregation}\label{sec:aggregation}
Once decisions are collective, constraint preferences raise an aggregation
problem. Different individuals may have different reference beliefs, or different information technologies determining which departures they can entertain. A planner must then aggregate not only how individuals value outcomes but also the probabilistic models that enter their evaluations. Let $I:=\{1,\dots,n\}$ denote the individuals and index the planner by $0$. For every $i=0,\dots,n$, let $\succsim_i$ be a maxmin preference relation represented by a pair $(u_i,P_i)$, where $u_i:X\to\R$ is nonconstant and affine, and the set $P_i\subseteq\Delta$ is nonempty, compact, and convex.

Although $\succsim_i$ is complete, its belief set also induces a robust comparison relation. Let $\succsim_i^U$ denote the \emph{unambiguous comparison} associated with the maxmin preference $\succsim_i$.
\begin{definition}[Unambiguous maxmin comparison]\label{def:umcomparison}\normalfont
For $f,g\in F$,
$$f\succsim_i^U g
        \quad\Longleftrightarrow\quad
        \E_p[u_i(f)]
        \geq
        \E_p[u_i(g)]
        \quad\forall p\in P_i.$$
\end{definition}

The relation $\succsim_i^U$, axiomatized by \citet{bew02}, is incomplete and records comparisons that are supported by every belief that enters individual $i$'s maxmin preference. 

\par Our aggregation result will require individuals' tastes to be sufficiently diverse in the following sense. We say that 
the profile $(\succsim_i)_{i\in I}$ satisfies \emph{c-diversity} if, for every $i\in I$, there exist outcomes $x_i,y_i\in X$ such that
$x_i\succ_i y_i$ and
        $x_i\sim_j y_i$ $\forall j\neq i.$ C-diversity, also known as ``independent prospects,'' is equivalent to the individual utility functions being linearly independent when $X$ is at least $n$-dimensional. This property is popular in the literature \citep[][]{weymark91,mongin98,robust16,zuber16,nak26}.

\begin{definition}[Unambiguous maxmin Pareto dominance]\label{def:umpareto}\normalfont
The social maxmin preference relation $\succsim_0$ satisfies \emph{unambiguous maxmin Pareto dominance} with respect to $(\succsim_i)_{i\in I}$ if, for every $f,g\in F$,
$f\succsim_i^U g$ $\forall i\in I$ implies
        $f\succsim_0 g.$
\end{definition}

The principle requires social respect only for comparisons that are robust to every model entering each individual's maxmin preference, so it restricts the ordinary Pareto requirement to a smaller class of unanimous comparisons. Section \ref{sec:pareto_comparison} compares it to existing principles.

\begin{theorem}[Aggregation of maxmin preferences]\label{thm:agg}
Let $(\succsim_i)_{i=0}^n$ be maxmin preference relations represented by $\{(u_i,P_i)\}_{i=0}^n$, and suppose the individual profile $(\succsim_i)_{i\in I}$ satisfies c-diversity. Then, $\succsim_0$ satisfies unambiguous maxmin Pareto dominance with respect to $(\succsim_i)_{i\in I}$ if and only if there exist a nonzero vector $\varphi\in\R^n_+$ and a constant $\gamma\in\R$ such that
\begin{equation}\label{eq:agg}
        u_0
        =
        \sum_{i\in I}\varphi_i u_i+\gamma
        \qquad\text{and}\qquad
        P_0
        \subseteq
        \bigcap_{i\in I:\varphi_i>0}P_i.
\end{equation}
\end{theorem}

\par Theorem \ref{thm:agg} separates two distinct determinants of collective beliefs. The Pareto principle determines the aggregation rule: every social belief must satisfy the constraints of all individuals receiving positive social weight. However, it does not determine the geometry of those constraints. Their geometry is instead determined by the process with which individuals form beliefs. This division of labor is what makes information technologies economically consequential. Holding the normative aggregation principle fixed, changing how individuals process information  changes the admissible collective beliefs and can therefore reverse social rankings of uncertain policies. Section \ref{sec:bregman} studies the implications of this technology channel.

\subsection{Comparison with other Pareto principles}\label{sec:pareto_comparison}

Unambiguous maxmin Pareto dominance is deliberately formulated directly on maxmin preferences. This distinguishes it from the revealed unambiguous Pareto principle of \citet[][Section IV.B]{robust16}. The revealed unambiguous relation associated with a maxmin preference represented by $(u_i,P_i)$ is then precisely the Bewley preference $\succsim_i^U$ in Definition \ref{def:umcomparison}.

Specialized to maxmin preferences, revealed unambiguous Pareto dominance requires
\begin{equation}\label{eq:dght_pareto}
        f\succsim_i^U g
        \quad\forall i\in I
        \qquad\Longrightarrow\qquad
        f\succsim_0^U g.
\end{equation}
By contrast, Definition \ref{def:umpareto} requires only
$f\succsim_i^U g
        \hspace{0.1in}\forall i\in I
        \Longrightarrow
        f\succsim_0 g.$
Thus, the requirement on the planner is weaker: \eqref{eq:dght_pareto} requires every belief in $P_0$ to support the unanimous comparison, whereas unambiguous maxmin Pareto dominance requires only the planner's complete maxmin ranking to support it. The next result shows that the weaker planner-side requirement suffices to produce the same implication on the domain of Theorem \ref{thm:agg}.

\begin{obs}\label{prop:pareto_comparison}
Under c-diversity, the social maxmin preference satisfies unambiguous maxmin Pareto dominance if and only if it satisfies revealed unambiguous Pareto dominance.
\end{obs}

Observation \ref{prop:pareto_comparison} clarifies the gain from working directly with maxmin preferences. Once c-diversity is imposed, our weaker requirement already forces the utility aggregation and intersection restriction that make the planner's unanimous comparisons robust for every belief. Thus, \citeauthor{robust16}'s (\citeyear{robust16}) revealed unambiguous social relation emerges as a consequence rather than being imposed as the object to which Pareto applies.

Another relevant comparison is \citet[][Section 3.3]{qu17}, who also studies the aggregation
of maxmin preferences. Assuming all individuals  agree on
the ranking of two outcomes,\footnote{This assumption is also known as ``c-minimal agreement'' and is implied by c-diversity.} he combines a Pareto requirement on outcomes with a restricted certainty-equivalent Pareto condition: for arbitrary finite collections of acts, the planner must preserve
unanimous comparisons between mixtures of  individual certainty equivalents. The implied aggregation rule is a
weighted Minkowski sum of individual belief sets, rather than their intersection. Although the two aggregation principles use different domains, our Pareto requirement has a simpler structure:
it requires neither Pareto on outcomes, the construction of certainty equivalents, nor unanimity comparisons
between their arbitrary mixtures.

\section{Information Technologies and Belief Aggregation}\label{sec:bregman}

Theorem \ref{thm:agg} reduces the belief side of the social problem to an intersection of individual constraint sets. The economic content of this intersection depends on the geometry of those sets. Section \ref{sec:strict_ri} showed that uniformly posterior-separable information costs make strict rational inattention generate dual Bregman balls, while Section \ref{sec:primal_framework} introduced  primal balls. This section studies how these geometries translate individual beliefs into collective beliefs. The main point is that our Pareto principle can produce different belief-aggregation rules because different information technologies generate different constraint sets.

\subsection{Linear pooling}\label{sec:pos}

Linear pooling is perhaps the most prominent aggregation rule, dating back to \citet{stone61}. The first result uses primal Bregman balls to partially identify a linear pooling of individuals' reference beliefs. For all results in this section, we impose the regularity conditions in Appendix \ref{app:conditions} and maintain a nonempty intersection of individual belief sets.

\begin{proposition}\label{thm:pos}
  Fix any generator $G$ and reference beliefs $(p_{i})_{i\in I}$ satisfying the regularity conditions in Appendix \ref{app:conditions} and assume $\bigcap_{i\in I}B^G_{\eta_i}(p_i)\neq\varnothing$.  Then, $\bigcap_{i\in I}B^G_{\eta_i}(p_i)\cap\text{\normalfont co}(p_{i})_{i\in I}\neq\varnothing$.
\end{proposition}
This result shows that if every  individual has a primal ball, then the intersection of these balls must contain at least one linear pooling of their centers; here, $\co(.)$ denotes the convex hull. The next two results build on this insight to recover familiar welfare criteria.

\begin{corollary}\label{thm:neutral}
There exists $(\eta_i)_{i\in I}$ such that
$\co(p_i)_{i\in I}
        \subseteq
        \bigcap_{i\in I}B^G_{\eta_i}(p_i)$.
\end{corollary}

Corollary \ref{thm:neutral} explains how familiar welfare criteria based on individuals with subjective expected utility (SEU) preferences can arise within our maxmin-based framework. In SEU-based criteria, every individual SEU preference is represented by a unique belief; in our framework, the same role is played by the reference beliefs $(p_i)_{i\in I}$, while individual preferences are maxmin whose belief sets are centered around $(p_i)_{i\in I}$. Thus, when the Bregman radii $(\eta_i)_{i\in I}$ are large enough that the intersection contains $\co(p_i)_{i\in I}$ and  $P_0\subseteq\co(p_i)_{i\in I}$ holds in Theorem \ref{thm:agg}, the resulting criterion coincides with the maxmin criterion in \citet[][Theorem 1]{util16}. When $P_0=\co(p_i)_{i\in I}$, the resulting criterion coincides with the maxmin criterion in \citet[][Theorem 2]{util16}; similarly, $\co(p_i)_{i\in I}$ is also the social-belief set in  \citeauthor{simsek14}'s (\citeyear{simsek14}) belief-neutral criterion. 

Notice that the social preference in Theorem \ref{thm:agg} is forced to be an SEU welfare criterion when the intersection of individual belief sets is a singleton; this is the most studied class of welfare criteria in the aggregation literature. Thus, for  our remaining analysis, we focus on singleton intersections to recover familiar SEU welfare criteria, then Section \ref{sec:nonsingleton} will generalize the analysis to non-singleton cases and provide an institutional interpretation for the singleton case. Corollary \ref{thm:pos1} point-identifies the most popular welfare criterion.

\begin{corollary}\label{thm:pos1}
If
$\bigcap_{i\in I}B^G_{\eta_i}(p_i)=\{q_0^*\}$,
then there exist convex weights $(\mu_i)_{i\in I}$ such that
$$q_0^*(s)
        =
        \sum_{i\in I}\mu_i p_i(s)
        \qquad\forall s\in S.$$
\end{corollary}

When the intersection of primal Bregman balls is a singleton, the unique point-identified belief in the intersection is always a linear pooling of reference beliefs---the model weights $(\mu_i)_{i\in I}$ depend on $G$ but we always suppress that dependence to ease notation. Combined with Theorem \ref{thm:agg}, Corollary \ref{thm:pos1} shows that the planner behaves identically to \citeauthor{gil04}'s (\citeyear{gil04}) SEU planner with social preference $(u_0,q^*_0)$,  
which is known as a ``utilitarian'' preference 
or ``linear-linear'' aggregation rule. Thus, the primal Bregman geometry recovers the belief component of \citeauthor{gil04}'s (\citeyear{gil04}) welfare criterion within a maxmin-preference environment: the key is that the reference beliefs $(p_i)_{i\in I}$ play the role of individual SEU beliefs, and the singleton intersection selects their linear pool as the unique collective belief. 

\subsection{Geometric pooling}\label{sec:geometric}

The previous analysis showed that the primal order identifies linear pooling under the intersection rule. The rest of our analysis builds on the dual order because, as we shall see, it is more flexible. We start with the case where individuals have dual entropy balls $\hat{B}^{\text{\normalfont\tiny KL}}_{\eta_i}(p_i)$. 

\begin{proposition}\label{thm:geo1}
If $\bigcap_{i\in I}\hat{B}^{\text{\normalfont\tiny KL}}_{\eta_i}(p_i)=\{\hat{q}_0\}$, then there exist convex weights $(\lambda_i)_{i\in I}$ such that
\begin{equation}\label{eq:geopool}
        \hat{q}_0(s)=
        \frac{\prod_{i\in I}p_i(s)^{\lambda_i}}
        {\sum_{t\in S}\prod_{i\in I}p_i(t)^{\lambda_i}} \qquad \forall s\in S.
\end{equation}
\end{proposition}
When the intersection of dual entropy balls is a singleton, the unique point-identified belief is always a \textit{geometric} (or logarithmic) pooling of reference beliefs. Combined with Theorem \ref{thm:agg}, Proposition \ref{thm:geo1} shows that our planner behaves identically to \citeauthor{diet21}'s (\citeyear{diet21}) SEU planner with social preference $(u_0,\hat{q}_0)$, which is known as a ``linear-geometric'' aggregation rule. \citet{diet21} axiomatizes this rule and shows that, under basic conditions, it is the only preference that is dynamically consistent, i.e., it is compatible with Bayesian updating.

\subsection{Power pooling and other aggregation rules}\label{sec:power}

We now consider a broader class of dual Bregman balls to show that they can also rationalize more sophisticated pooling rules. Here, we will obtain different belief-aggregation rules by varying the generator $G$ while holding the dual order fixed. To this end, let $G$ be \textit{separable}, in the sense that
$G(q)=\sum_{s\in S}\beta_s g(q(s)),$
where $\beta_s>0$ for every $s\in S$, and  $g$ is twice continuously differentiable and strictly convex on $(0,\infty)$, with $g''>0$. 

\begin{proposition}\label{prop:separable_pool}
 If
$\bigcap_{i\in I}\hat B^G_{\eta_i}(p_i)=\{\tilde{q}_0\}$ and $G$ is separable,
then there exist convex weights $(\mu_i)_{i\in I}$ and a constant $a\in\mathbb{R}$ such that
$\beta_s g'(\tilde{q}_0(s))=a+\beta_s\sum_{i\in I}\mu_i g'(p_i(s))$ $\forall s\in S.$
\end{proposition}

This formula is very flexible. If $\beta_s$ is constant across states and  $g'(x)=x$, then we recover the linear pooling in Corollary \ref{thm:pos1}, and $g'(x)=\log x$ recovers the geometric pooling in Proposition \ref{thm:geo1}. With general $\beta_s$, if $g'(x)=x^r/r$ for $r\neq0$, then the formula becomes
\begin{align}\label{eq:power}
     \tilde{q}_0(s)=
        \Big(\frac{ra}{\beta_s}+\sum_{i\in I}\mu_i p_i(s)^r\Big)^{1/r}
        \qquad
        \forall s\in S,
\end{align}
where $a\in\R$ is a normalizing constant that ensures $\tilde{q}_0\in\Delta$. Here, the parameter $r$ controls how the planner treats disagreement across states. Larger $r$ puts relatively more weight on reference beliefs that assign high probability to a state, whereas smaller $r$ gives relatively more influence to reference beliefs that assign low probability to that state.  

\par Notably, the expression \eqref{eq:power} also nests the standard \textit{power} pooling rule 
        $$
      \tilde{q}_0(s)=  \frac{\big(\sum_{i\in I}\mu_i p_i(s)^r\big)^{1/r}}
        {\sum_{t\in S}\big(\sum_{i\in I}\mu_i p_i(t)^r\big)^{1/r}}
        \qquad
        \forall s\in S.
$$
when $a=(Z^{-r}-1)/r$, $\beta_s=1/V_s$, $Z:=\sum_{t\in S}V_t^{1/r}$, and  $V_s:=\sum_{i\in I}\mu_i p_i(s)^r$ in \eqref{eq:power}. This rule is used in welfare economics to measure multi-dimensional inequality \citep[e.g.,][]{maa86}, and it is also used in finance for asset pricing \citep[e.g.,][]{asset21}. 
\par The next result identifies another relevant aggregation rule related to geometric pooling.

\begin{corollary}\label{thm:multi}
There exist $(p_i)_{i\in I}$, $(\eta_i)_{i\in I}$, and separable $G$  such that $\bigcap_{i\in I}\hat B^G_{\eta_i}(p_i)=\{q^{\times}_0\}$  and
$$
        q^{\times}_0(s)=
        \frac{\prod_{i\in I}p_i(s)}
        {\sum_{t\in S}\prod_{i\in I}p_i(t)}
        \qquad \forall s\in S.
$$
\end{corollary}

This rule is known  as \textit{multiplicative} (or product) pooling and is axiomatized by \citet{diet10}, who relates it to aggregation of private information. Its key property is that it is parameter-free. However, unlike all previous results, Corollary \ref{thm:multi} requires reference beliefs to satisfy a ``compatibility'' restriction because this rule does not satisfy the unanimity axiom.\footnote{If $\varPhi:\Delta^{|I|}\rightarrow\Delta$ is a belief-aggregation rule, the unanimity axiom requires $\varPhi(p,\dots,p)=p$. Now, suppose $p_i=p$  $\forall i\in I$ and $|I|>1$. The multiplicative rule yields   $q^\times(s)=p(s)^{|I|}/\sum_{t\in S}p(t)^{|I|}$, which differs from the common $p$, except when $p$ is uniform. In contrast, all other rules above satisfy unanimity.} 
\par To summarize, all our previous results have shown that changing the Bregman geometry changes the economic meaning of belief aggregation: it changes whether collective belief formation emphasizes optimistic  beliefs, pessimistic beliefs, or more balanced beliefs.

\subsection{Beyond singleton intersections and interpretations}\label{sec:nonsingleton}

The preceding results point-identify collective beliefs. Non-singleton intersections are also useful because they permit several pooling rules to coexist in a single social constraint set.

\begin{obs}\label{thm:gen}
Let $\mathcal A$ be the closed convex hull of the collection of beliefs generated by the linear, geometric, power, and multiplicative pooling rules of the reference beliefs $(p_i)_{i\in I}$. Then, there exist $(\eta_i)_{i\in I}$ such that
$\mathcal A
        \subseteq
        \bigcap_{i\in I}
        \hat B^{\text{\normalfont\tiny KL}}_{\eta_i}(p_i).$
\end{obs}

 The non-singleton case therefore has economic content rather than representing a failure of point identification. Different environments can justify different ways of pooling the same reference beliefs, and a sufficiently rich social constraint region can accommodate several such beliefs simultaneously. This aspect of our framework is consistent with the epistemic social choice literature \citep[e.g.,][]{epistemic16,pivato19,update20,pivato22,pivato24}. The planner's residual ambiguity then records the fact that the individual constraints do not select a unique admissible probabilistic model. 

The singleton condition  can be interpreted institutionally through the lens of social \textit{deliberation}. \citet{delib24} study deliberation as repeated interactions on a social network in which individuals exchange beliefs and iteratively update their beliefs according to some rule, such as \citeauthor{degroot74}'s (\citeyear{degroot74}) rule. They show that deliberation can converge to a consensus belief that takes various forms depending on the updating rule, such as  the linear pooling  in Corollary \ref{thm:pos1} and the geometric pooling in Proposition \ref{thm:geo1}. This suggests interpreting the singleton condition  as a reduced-form outcome of a convergent deliberation process: individuals start with their own reference beliefs, agree to exchange beliefs with others, and update until a consensus belief is reached. When this deliberation process converges, the consensus belief can be interpreted as the unique collective belief in the singleton intersection.

\section{Applications}\label{sec:app}
We explore three applications of our framework. Section \ref{sec:app_finance} considers financial markets, Section \ref{sec:app_infoagg} analyzes prediction markets, and Section \ref{sec:app_mechanism} studies mechanism design.
\subsection{Financial markets}\label{sec:app_finance}

This first application shows how aggregation of constraint preferences is reflected in observable equilibrium prices. The natural benchmark is \citet{gil00}, who study an exchange economy with maxmin expected utility and show that, in the absence of aggregate uncertainty, full insurance is Pareto optimal precisely when agents share at least one belief. We build on this environment to isolate the pricing implications of the Bregman geometry.

Consider a two-date Arrow-Debreu exchange economy. At date $1$, one state $s\in S$ realizes. There is one consumption good in every state and a complete set of Arrow securities traded at date $0$. Individual $i$ has endowment $\omega^i\in\R^S_{++}$, a differentiable, strictly increasing, concave utility $u_i:\R_+\to\R$ with $u_i'(z)>0$, and a nonempty closed convex belief set $P_i\subseteq\Delta$. Her preference over consumption bundles is represented by the maxmin criterion
$V_i(x^i)
        :=
        \min_{p\in P_i}
        \sum_{s\in S}p(s)u_i(x_s^i).$ Assume no aggregate uncertainty:
$\sum_{i\in I}\omega_s^i=\bar W$ $\forall s\in S$, where $\bar W>0$. An allocation $(x^i)_{i\in I}$ is feasible if
$\sum_{i\in I}x_s^i=\bar W$ $\forall s\in S$. A feasible allocation is \emph{full insurance} if $x_s^i=c_i$ for every state and every individual. A competitive equilibrium consists of a feasible allocation and a nonzero Arrow-price vector normalized on $\Delta$.

\begin{lemma}[Full insurance under maxmin preferences]\label{lem:finance_RS}
Suppose $\bigcap_{i\in I}P_i\neq\varnothing$.
\begin{enumerate}
\item Every feasible full-insurance allocation is Pareto optimal for the maxmin preferences represented by the profile $(V_i)_{i\in I}$.
\item For every $q\in\bigcap_{i\in I}P_i$, the allocation
$x_s^i\equiv q\cdot\omega^i$
is a competitive equilibrium supported by normalized Arrow price $q$.
\item Conversely, the normalized Arrow price of every full-insurance competitive equilibrium belongs to $\bigcap_{i\in I}P_i$. Hence, if $\bigcap_{i\in I}P_i=\{q^*\}$, every full-insurance equilibrium has the unique normalized Arrow price $q^*$.
\end{enumerate}
\end{lemma}

The first statement captures the full-insurance result in \citet[][Theorem 1]{gil00}: a common prior is enough to rule out a Pareto improvement from a fully insured allocation. The second and third statements add the price implication needed here. Every shared prior can support an appropriately distributed full-insurance allocation, and every price supporting any full-insurance equilibrium must itself be shared. Thus, a singleton intersection turns the belief selected by the aggregation principle into an observable state-price kernel.

\par For comparison, Lemma \ref{prop:finance_risk_benchmarks}  presents two classical complete-preference benchmarks that aggregate heterogeneous priors in different ways. Throughout, every individual $i$ has a full-support reference belief $p_i\in P_i$, which later will represent the center of their Bregman ball.

\begin{lemma}\label{prop:finance_risk_benchmarks}
Let $J\subseteq I$ be nonempty and let $(\mu_i)_{i\in J}$ be positive convex weights.
\begin{enumerate}
\item[(i)] The log-utility planner problem
$$
        \max_{(x^i)_{i\in J}}
        \sum_{i\in J}\mu_i
        \sum_{s\in S}p_i(s)\log x_s^i
        \qquad
        \text{s.t.}\quad
        \sum_{i\in J}x_s^i=\bar W
        \quad\forall s
$$
is supported by the normalized Arrow price
$q^{\mathrm L}(s)=\sum_{i\in J}\mu_i p_i(s)$.
\item[(ii)] Fix $\tau>0$, set $\tau_i=\tau\mu_i$, and let consumption be unrestricted in $\R$. The exponential-utility planner problem
$$
        \max_{(x^i)_{i\in J}}
        \sum_{i\in J}\alpha_i
        \sum_{s\in S}p_i(s)
        (-e^{-x_s^i/\tau_i})
        \qquad
        \text{s.t.}\quad
        \sum_{i\in J}x_s^i=\bar W
        \quad\forall s,
$$
where $\alpha_i>0$, is supported by the normalized Arrow price
$$
        q^{\mathrm H}(s)
        =
        \frac{\prod_{i\in J}p_i(s)^{\mu_i}}
        {\sum_{t\in S}\prod_{i\in J}p_i(t)^{\mu_i}}.
$$
\end{enumerate}
\end{lemma}
Lemma \ref{prop:finance_risk_benchmarks} formalizes how two classical risk economies price beliefs in the absence of aggregate uncertainty. In the log-utility benchmark, normalized Arrow prices average probability levels across individuals. In the exponential-utility benchmark, normalized Arrow prices average log probabilities across individuals. These benchmarks differ from the maxmin economy, where individuals have belief sets, and prices are disciplined by shared beliefs. 

\subsubsection{Market implications of Bregman geometry}\label{sec:finance_bregman}
We now specialize the full-insurance pricing result to entropy balls, where the only difference between the two cases is the order of arguments in the underlying relative entropy.
\begin{proposition}\label{thm:finance_entropy_order}
Assume no aggregate uncertainty and consider full-insurance equilibria.
\begin{enumerate}
\item[(i)] Suppose
$P_i=B^{\text{\normalfont\tiny KL}}_{\eta_i}(p_i)$
for every $i$ and
$\bigcap_iP_i=\{q^{\mathrm L}\}$.
Then, there are a nonempty $J_{\mathrm L}\subseteq I$ and positive convex weights $(\mu_i^{\mathrm L})_{i\in J_{\mathrm L}}$ such that
$$
        q^{\mathrm L}(s)
        =
        \sum_{i\in J_{\mathrm L}}
        \mu_i^{\mathrm L}p_i(s).
$$
Every full-insurance maxmin equilibrium has the unique normalized Arrow price $q^{\mathrm L}$, which is also the price in the log-utility economy of Lemma \ref{prop:finance_risk_benchmarks}.(i).
\item[(ii)] Suppose
$P_i=\hat B^{\text{\normalfont\tiny KL}}_{\eta_i}(p_i)$
for every $i$ and
$\bigcap_iP_i=\{q^{\mathrm H}\}$.
Then, there are a nonempty $J_{\mathrm H}\subseteq I$ and positive convex weights $(\mu_i^{\mathrm H})_{i\in J_{\mathrm H}}$ such that
$$
        q^{\mathrm H}(s)
        =
        \frac{\prod_{i\in J_{\mathrm H}}
        p_i(s)^{\mu_i^{\mathrm H}}}
        {\sum_{t\in S}\prod_{i\in J_{\mathrm H}}
        p_i(t)^{\mu_i^{\mathrm H}}}.
$$
Every full-insurance maxmin equilibrium has the unique normalized Arrow price $q^{\mathrm H}$, which is also the price in the exponential-utility economy of Lemma \ref{prop:finance_risk_benchmarks}.(ii).
\end{enumerate}
\end{proposition}

This result gives the Bregman geometry a pricing interpretation in financial markets. The order of arguments selects which shared belief is capitalized into equilibrium prices in full-insurance maxmin economies. Specifically, primal entropy yields the prices of a log-utility risk economy, whereas dual entropy yields the prices of an exponential-utility risk economy. Observation \ref{prop:finance_power_crra} extends the same logic to power pooling and isoelastic risk sharing.

\subsubsection{Illustration: financial losses and information technologies}\label{sec:pricing_disagreement}

Section \ref{sec:bregman} indicated that primal entropy aggregates probability levels, whereas dual entropy aggregates log probabilities. Proposition \ref{thm:finance_entropy_order} shows that this distinction has an asset-pricing implication: using the wrong entropy order is equivalent to using the prices of a different risk economy. We now illustrate how this mistake generates financial losses.

Consider two states, $S=\{A,B\}$, and two individuals, $I=\{1,2\}$, with reference beliefs $p_1=(0.96,0.04)$ and $p_2=(0.36,0.64)$. Individual $1$ views state $B$ as very unlikely, while individual $2$ views it as fairly likely. Suppose the information technology is the dual-entropy technology, so individual belief sets are dual entropy balls, $P_i=\hat B^{\text{\normalfont\tiny KL}}_{\eta_i}(p_i)$. Suppose the radii satisfy $\eta_1=\KL(q^{\mathrm H}\lVert p_1)=0.202$ and $\eta_2=\KL(q^{\mathrm H}\lVert p_2)=0.379$, where $q^{\mathrm H}=(0.786,0.214)$. This belief is the equal-weight geometric pool of $p_1$ and $p_2$: its odds ratio is the geometric average of the two individual odds ratios, since $q^{\mathrm H}(A)/q^{\mathrm H}(B)=\sqrt{(p_1(A)/p_1(B))(p_2(A)/p_2(B))}=\sqrt{13.5}$. Moreover, it follows that the intersection $P_1\cap P_2=\{q^{\mathrm H}\}$ is a singleton. By Proposition \ref{thm:finance_entropy_order}, every full-insurance maxmin equilibrium has the unique normalized Arrow price vector $q^{\mathrm H}$; in particular, the state-$B$ Arrow security has price $q^{\mathrm H}(B)=0.214$. 

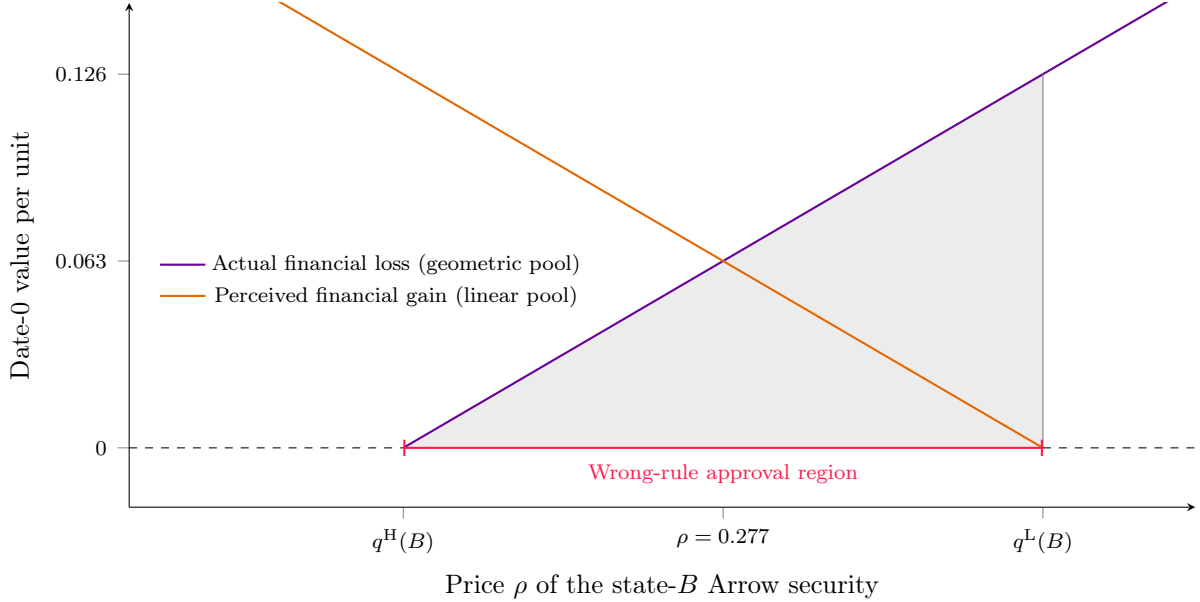
\begin{figure}[hbt!]
\centering
\begin{tikzpicture}
\begin{axis}[
    width=0.95\textwidth,
    height=0.5\textwidth,
    xmin=0.16, xmax=0.37,
    ymin=-0.02, ymax=0.15,
    xlabel={Price $\rho$ of the state-$B$ Arrow security},
    ylabel={Date-$0$ value per unit},
    axis lines=left,
    label style={font=\footnotesize},
    tick label style={font=\scriptsize},
    legend style={at={(0.02,0.45)},anchor=west,draw=none,fill=none,font=\scriptsize},
    tick align=outside,
    xtick={0.214,0.277,0.34},
    xticklabels={$q^{\mathrm H}(B)$,$\rho=0.277$,$q^{\mathrm L}(B)$},
    ytick={0,0.063,0.126},
    yticklabels={$0$,$0.063$,$0.126$},
]

\addplot[name path=loss, thick, darkpurple, domain=0.214:0.34, samples=2] {x-0.214};
\addplot[name path=gain, thick, darkorange, domain=0.214:0.34, samples=2] {0.34-x};
\addplot[name path=xaxis, draw=none, domain=0.214:0.34] {0};

\addplot[fill=gray!15] fill between[of=loss and xaxis];

\draw[gray, thin] (axis cs:0.34,0) -- (axis cs:0.34,0.126);

\addplot[thick, darkpurple, domain=0.34:0.37, samples=2] {x-0.214};
\addplot[thick, darkorange, domain=0.16:0.214, samples=2] {0.34-x};

\addlegendentry{Actual financial loss (geometric pool)}
\addlegendentry{Perceived financial gain (linear pool)}

\addplot[dashed, domain=0.16:0.214, samples=2] {0};
\addplot[dashed, domain=0.34:0.37, samples=2] {0};

\draw[thick, darkred, |-|] (axis cs:0.214,0) -- (axis cs:0.34,0);

\node[anchor=north, font=\scriptsize, text=darkred] at (axis cs:0.277,-0.002) {Wrong-rule approval region};

\end{axis}
\end{tikzpicture}
\caption{Financial losses and information technologies.}
\label{fig:mispricing_loss}
\end{figure}

\par Suppose a financial authority evaluates whether buying one unit of state-$B$ security at date-$0$ price $\rho$ is socially valuable. If the authority ignores the dual-entropy technology and instead uses a linear pooling, it acts as if the relevant benchmark were the log-utility economy in Lemma \ref{prop:finance_risk_benchmarks}.(i), rather than the exponential-utility economy in Lemma \ref{prop:finance_risk_benchmarks}.(ii). The equal-weight linear pooling is $q^{\mathrm L}=(0.66,0.34)$, so the mistaken authority values the state-$B$ security at $q^{\mathrm L}(B)=0.34$. Figure \ref{fig:mispricing_loss} illustrates the financial loss caused by applying the incorrect aggregation rule. The ``wrong-rule approval region'' $(0.214 < \rho < 0.34)$ denotes prices where the linear-pooling authority incorrectly evaluates the security as underpriced and authorizes a purchase. Within this region, the downward-sloping line ($0.34-\rho$) plots the authority's perceived gain per unit, while the upward-sloping line ($\rho-0.214$) plots the actual financial loss evaluated under the correct dual-entropy technology. The intersection at $\rho=0.277$ marks the price where the perceived gain equals the actual loss. To the left of this point ($0.214 < \rho < 0.277$), the perceived gain exceeds the actual loss, leading the authority to view the trade as favorable despite it reducing overall value. To the right ($0.277 < \rho < 0.34$), the actual loss surpasses the perceived gain. The shaded region traces the actual per unit loss across prices at which the authority approves the trade. 

\subsubsection{Power pooling and risk aversion}\label{sec:finance_power}
The entropy comparison isolates the effect of changing the order of relative entropy. The same pricing logic extends to the broader separable Bregman family from Section \ref{sec:power}. When the singleton intersection is a power pool, the resulting state-price kernel has a standard risk-sharing interpretation: it is the Arrow price vector generated by an \textit{isoelastic} risk (CRRA) economy. Thus, the power parameter $r$ in \eqref{eq:power} has an economic meaning in financial markets.
\begin{obs}\label{prop:finance_power_crra}
Maintain the assumptions of Lemma \ref{lem:finance_RS}. Fix $r>0$, a nonempty $J_{\mathrm P}\subseteq I$, and positive convex weights $(\mu_i)_{i\in J_{\mathrm P}}$. Suppose the separable dual Bregman construction in Proposition \ref{prop:separable_pool} yields
$\bigcap_iP_i=\{q_r^{\mathrm P}\}$,
where
$$
q_r^{\mathrm P}(s)
=
\frac{\big(\sum_{i\in J_{\mathrm P}}\mu_i p_i(s)^r\big)^{1/r}}
{\sum_{t\in S}\big(\sum_{i\in J_{\mathrm P}}\mu_i p_i(t)^r\big)^{1/r}}.
$$
Then, every full-insurance maxmin equilibrium has the unique normalized Arrow price $q_r^{\mathrm P}$. This is also the normalized Arrow price of the risk economy whose planner solves
$$
        \max_{(x^i)_{i\in J_{\mathrm P}}}
        \sum_{i\in J_{\mathrm P}}
        \mu_i^{1/r}
        \sum_{s\in S}p_i(s)u_r(x_s^i)
        \qquad
        \text{s.t.}\quad
        \sum_{i\in J_{\mathrm P}}x_s^i=\bar W
        \quad\forall s,
$$
where $u_1(z)=\log z$ and
$u_r(z)=z^{1-1/r}/(1-1/r)$ for $r\neq1$. Thus, $r$ corresponds to an isoelastic economy with relative risk aversion $1/r$.
\end{obs}

\subsection{Information aggregation in prediction markets}\label{sec:app_infoagg}

The second application studies how the structure of constraint preferences affects the aggregation of dispersed information in prediction markets. \citet{agg24} show that the separable securities that aggregate information under precise expected-utility beliefs may fail when traders have maxmin preferences. The reason is inertia: if the inherited market announcement is the expectation of the security under one posterior in a trader's belief set, the trader can optimally repeat that announcement. They introduce strongly separable securities and show that strong separability characterizes robust information aggregation.\footnote{The published version \citet{agg24} assumes a common belief set. Their 2019 working-paper version notes that their strategic proof can accommodate different belief sets with a nonempty intersection. We do not rely on that extension because our analysis is derived trader by trader from each belief set.}

Specializing their arbitrary convex belief sets to dual entropy balls delivers new economic insights. Notably, we show that it determines  how ambiguity changes after a private signal,  how much ex-ante capacity is required to rationalize a given market price, and  how the required capacity depends on the rarity of the information event. 

\subsubsection{Market-scoring-rule environment}\label{sec:infoagg_setup}

Let $Z:S\to\R$ be the intrinsic payoff of a security, and let
$\mathcal Y:=\left[\min_{s\in S}Z(s),\max_{s\in S}Z(s)\right]$
be the set of admissible announcements. Trader $i$ has information partition $\Pi_i$, and the join of the partitions reveals the state. At any trading time, the trader observes her current information cell and the preceding public announcement $v\in\mathcal Y$. If she announces $y\in\mathcal Y$, her incremental payoff is
$\mathsf s(y,Z)-\mathsf s(v,Z)$,
where $\mathsf s$ is a continuous strictly proper scoring rule on $\mathcal Y$ for the expectation of $Z$: for every probability $r$, the map
$y\mapsto\E_r[\mathsf s(y,Z)]$
is uniquely maximized at $y=\E_r[Z]$. Each trader has a full-support reference belief $p_i\in\Delta^\circ$ and dual entropy constraint set
$P_i(\eta_i):=\hat B^{\text{\normalfont\tiny KL}}_{\eta_i}(p_i)$.
For any nonempty event $E$ with $p_i(E)>0$, define the conditional image
$P_i^E(\eta_i)
        :=
        \big\{q(\cdot|E):q\in P_i(\eta_i), q(E)>0\big\}.
$

Information aggregation means that the sequence of market announcements converges in probability to the intrinsic-value random variable $Z$, as in \citet{agg24}. For reference, strong separability has a purely informational formulation. A security is strongly separable under the partition structure if, for every $v\in\R$ and every nonempty $E\subseteq\{s:Z(s)\neq v\}$, there exist a trader $i$, a state $s\in E$, and $\lambda\in\R$ such that
$(Z(t)-v)\lambda>0$
        $\forall t\in\Pi_i(s)\cap E.$
Thus, some trader can distinguish the inherited value $v$ from the security's payoff in a common direction on some information cell. \citet{agg24} show that this condition is exactly what is needed for robust information aggregation in their maxmin prediction-market environment. We take that qualitative benchmark as given and use dual entropy to study quantitatively what happens when strong separability does not itself resolve the market.

\begin{lemma}[Conditional entropy constraint]\label{prop:conditional_KL}
Fix trader $i$ and a nonempty event $E$. Let $a_i:=p_i(E)$ and $p_i^E:=p_i(\cdot|E)$. If $0<a_i<1$, then
\begin{equation}\label{eq:conditional_KL}
P_i^E(\eta_i)
=
\begin{cases}
\big\{r\in\Delta(E):\KL(r\lVert p_i^E)\leq\rho_i(E,\eta_i)\big\},
&
\eta_i<-\log(1-a_i),\\[0.05in]
\Delta(E),
&
\eta_i\geq-\log(1-a_i),
\end{cases}
\end{equation}
where
\begin{equation}\label{eq:rho}
        \rho_i(E,\eta_i)
        :=
        -\log\frac{e^{-\eta_i}-1+a_i}{a_i}.
\end{equation}
If $a_i=1$, then
$P_i^E(\eta_i)
        =
        \big\{r\in\Delta(E):\KL(r\lVert p_i^E)\leq\eta_i\big\}.$
\end{lemma}

Lemma \ref{prop:conditional_KL} is the key dynamic implication of the dual-entropy structure. When $\eta_i<-\log(1-p_i(E))$, the conditional ambiguity set in \eqref{eq:conditional_KL} remains an entropy ball, but its radius becomes $\rho_i(E,\eta_i)$. This radius is increasing in the ex-ante capacity $\eta_i$ and decreasing in the reference probability $p_i(E)$. Thus, the same ex-ante information-processing capacity generates a larger conditional ambiguity set after a more surprising signal.

The second branch of \eqref{eq:conditional_KL} reveals something stronger. Once
$\eta_i\geq-\log(1-p_i(E))$,
conditioning on $E$ makes the belief set the entire simplex $\Delta(E)$. For a rare event,
$-\log(1-p_i(E))\simeq p_i(E)$.
Therefore, a small ex-ante capacity can generate complete conditional ambiguity following a sufficiently surprising event. This provides a sharp mechanism through which rare events can impair information aggregation even when ex-ante ambiguity is modest. The same calculation identifies when prior-by-prior updating is regular.

\begin{corollary}[Regular updating]\label{cor:infoagg_regular}
If $0<p_i(E)<1$ and
$\eta_i<-\log(1-p_i(E)),$
then every $q\in P_i(\eta_i)$ satisfies $q(E)>0$.
\end{corollary}

Now, let $M_i(E;\eta_i)
        :=
        \big\{\E_r[Z]:r\in P_i^E(\eta_i)\big\}$ be the set of entropy-consistent conditional valuations after event $E$.
Since $P_i^E(\eta_i)$ is compact and convex and expectation is linear, the set $M_i(E;\eta_i)$ is a compact interval.

\begin{lemma}\label{prop:absorbing_announcement}
Suppose the previous announcement is $v\in\mathcal Y$ and trader $i$ observes $E$. Repeating $v$ is the unique maxmin-optimal announcement if and only if
$v\in M_i(E;\eta_i)$.
\end{lemma}

Thus, inherited prices are sticky precisely when they remain rationalizable as expectations under the trader's conditional entropy set. This recovers the price-inertia mechanism in \citet{agg24}, while the entropy geometry allows us to identify the information capacity required to rationalize any particular inherited price.

\subsubsection{The information cost of rationalizing a price}\label{sec:pricecost}

For $v$ in the convex hull of the payoffs $\{Z(s):s\in E\}$, define
\begin{equation}\label{eq:Iprice}
        I_i(E,v)
        :=
        \min\left\{
        \KL(r\lVert p_i^E):
        r\in\Delta(E),\
        \E_r[Z]=v
        \right\}.
\end{equation}
If $v$ lies outside that convex hull, set $I_i(E,v)=+\infty$. Here, $I_i(E,v)$ measures the smallest conditional information distortion required to make $v$ an expected value of the security.

\begin{proposition}\label{prop:critical_capacity}
Fix $i$, $E$, and a value $v\in\text{\normalfont co}(Z(E))$. If $0<a_i:=p_i(E)<1$, define
\begin{equation}\label{eq:critical_eta}
        \eta_i^\star(E,v)
        :=
        -\log
        \left(
        1-a_i+a_i e^{-I_i(E,v)}
        \right).
\end{equation}
If $a_i=1$, define
$\eta_i^\star(E,v):=I_i(E,v)$.
Then, $v\in M_i(E;\eta_i)
        \Longleftrightarrow
        \eta_i\geq\eta_i^\star(E,v).$
\end{proposition}

Proposition \ref{prop:critical_capacity} turns price inertia into a capacity threshold. The quantity $I_i(E,v)$ measures how far the target price is from trader $i$'s reference valuation after the signal.  \eqref{eq:critical_eta} converts this conditional distortion into the ex-ante capacity required to sustain it. Thus, a price does not become sticky just because a trader is ambiguity averse. It becomes sticky when the trader's information-processing capacity is large enough to contain a posterior whose expectation exactly equals the inherited price.

The distortion in \eqref{eq:Iprice} has a closed-form representation. Let
$\psi_i^E(\lambda)
        :=
        \log
        \sum_{s\in E}
        p_i^E(s)e^{\lambda Z(s)}.$

\begin{obs}\label{prop:price_tilt}
Suppose $Z$ is nonconstant on $E$ and
$\min_{s\in E}Z(s)<v<\max_{s\in E}Z(s)$. Then, there exists a unique $\lambda_i(E,v)\in\R$ with
$(\psi_i^E)'(\lambda_i(E,v))=v$.
Moreover,
\begin{equation}\label{eq:ratefunction}
        I_i(E,v)
        =
        \lambda_i(E,v)v
        -
        \psi_i^E(\lambda_i(E,v))
        =
        \sup_{\lambda\in\R}
        \{\lambda v-\psi_i^E(\lambda)\},
\end{equation}
and the unique minimizing posterior in \eqref{eq:Iprice} is
\begin{equation}\label{eq:tiltedposterior}
        r_i^{E,v}(s)
        =
        \frac{p_i^E(s)e^{\lambda_i(E,v)Z(s)}}
        {\sum_{t\in E}p_i^E(t)e^{\lambda_i(E,v)Z(t)}}.
\end{equation}
\end{obs}

Thus, the cheapest way to rationalize an inherited price is an exponential tilt of the trader's conditional reference belief toward states in which the security payoff moves in the required direction. The magnitude $\lambda_i(E,v)$ is the shadow intensity of this distortion.  \eqref{eq:ratefunction} is economically useful because it maps an observed price directly into the minimum informational distortion required to support it. Moreover, combining Proposition \ref{prop:critical_capacity} and Observation \ref{prop:price_tilt} yields a sharp comparative static for surprising information.

\begin{corollary}[Rare-event fragility]\label{cor:rare_fragility}
Fix a conditional reference distribution $p^E$, a non-reference valuation $v$, and let $I:=I(E,v)>0$. Holding $p^E$ and $I$ fixed, define
$T(a,I)
        :=
        -\log(1-a+ae^{-I})$, for every $a\in(0,1)$.
Then, $T(a,I)$ is strictly increasing in $a$ and $\lim_{a\downarrow0}
        \frac{T(a,I)}{a}
        =
        1-e^{-I}.$
Thus, 
$T(a,I)
        =
        a(1-e^{-I})+o(a)$ as $a\downarrow0$.
\end{corollary}

The economic implication is as follows: conditional on needing the same posterior distortion $I$ to rationalize a price, a rarer information event requires a smaller ex-ante capacity to sustain that price. In the rare-event limit, the required capacity is only first order in the probability of the event. For small distortions $I$, $1-e^{-I}\simeq I$, so the threshold is approximately $aI$. Rare signals therefore magnify the effect of a given information-capacity budget in a precise sense. This comparative static sharpens the intuition behind prediction failures for unfamiliar events. A market price that would be impossible to rationalize after an ordinary signal may become entirely consistent with strict rational inattention after a sufficiently surprising one, even though the trader's ex-ante capacity has not changed.

\par The price-cost formula becomes particularly transparent for Arrow-type securities.

\begin{corollary}\label{cor:arrow_pricecost}
Let $A\subseteq E$ and let $Z=\mathds{1}_A$ on $E$. Suppose
$\pi_i:=p_i^E(A)\in(0,1)$.
For every $v\in[0,1]$, $I_i(E,v)
        =
        d(v\lVert\pi_i)
        :=
        v\log\frac{v}{\pi_i}
        +(1-v)\log\frac{1-v}{1-\pi_i}.$
Thus, if $0<p_i(E)<1$,
\begin{equation}\label{eq:arrowthreshold}
        \eta_i^\star(E,v)
        =
        -\log
        \left[
        1-p_i(E)
        +
        p_i(E)e^{-d(v\lVert\pi_i)}
        \right].
\end{equation}
\end{corollary}

For a binary prediction market, the informational discrepancy (or distance) between the inherited price and the reference conditional probability is the binary relative entropy.  \eqref{eq:arrowthreshold} then converts this familiar statistical discrepancy into the ex-ante information capacity needed to make the price self-confirming. This gives the model a measurable implication: holding the reference forecast fixed, more extreme mispricing requires more information-processing capacity; holding the mispricing fixed, a rarer signal requires less.

\subsubsection{A quantitative no-learning frontier}\label{sec:infoagg_nolearning}

Let $E$ be the current public event. If the true state is $s\in E$, trader $i$'s private information refines it to
$C_i(s;E)
        :=
        E\cap\Pi_i(s).$
Define the set of prices that every trader can rationalize after every private signal compatible with $E$:
$\mathcal A(E;\eta)
        :=
        \bigcap_{i\in I}
        \bigcap_{s\in E}
        M_i(C_i(s;E);\eta_i).$
For each trader define the critical capacity
$\bar\eta_i(E,v)
        :=
        \max_{s\in E}
        \eta_i^\star(C_i(s;E),v),$
with the convention $\bar\eta_i(E,v)=+\infty$ if $v$ is outside the convex hull of $Z(C_i(s;E))$ for some $s$.

\begin{theorem}[No-learning frontier]\label{prop:no_learning_frontier}
For every public event $E$ and announcement $v\in\mathcal Y$,
\begin{equation}\label{eq:nliff}
        v\in\mathcal A(E;\eta)
        \quad\Longleftrightarrow\quad
        \eta_i\geq\bar\eta_i(E,v)
        \quad\forall i\in I.
\end{equation}
If these inequalities hold and the true state $s^\star\in E$ satisfies $Z(s^\star)\neq v$, then in the myopic market-scoring-rule game every trader uniquely repeats $v$ whenever called, no additional information is revealed, and information does not aggregate at $s^\star$.
\end{theorem}

Theorem \ref{prop:no_learning_frontier} is the main implication of imposing the dual entropy structure on the prediction-market model. The qualitative multiple-priors logic says that a wrong price can become absorbing. The entropy structure gives an exact frontier in the  information capacities. A market maker can trap the market at $v$ precisely when every trader's capacity is above the largest threshold generated by her possible private information cells.

This result also makes the source of price inertia transparent. Increasing an entropy radius weakly expands $\mathcal A(E;\eta)$.  Corollary \ref{cor:rare_fragility} shows that sustaining a fixed conditional distortion requires less ex-ante capacity when the information event is rarer. Thus, for a given capacity, an inherited price that cannot be rationalized after a relatively common signal may become rationalizable after a sufficiently rare one. Since such a price is then optimally repeated, rare signals can enlarge the range of self-confirming prices and make no-learning outcomes easier to sustain. This provides a quantitative information-cost channel through which prediction markets can have greater difficulty aggregating information following unusual events.

\subsection{Mechanism design and information rents}\label{sec:app_mechanism}
 This application studies how information technology affects information rents and therefore the value of dynamic mechanism design. \citet{strack25} show that the revelation principle is tied to expected utility (EU) and that concave non-EU preferences, including maxmin preferences, impose type-specific lower bounds on the rents achievable by dynamic mechanisms. We show how the geometry of constraint preferences determines these bounds and, under dual entropy, how the screening gain varies with information capacity.

We build on the Anscombe-Aumann extension in \citet[][Section 6.4]{strack25}. Fix a finite consequence set $X_0\subseteq X$ and let $\mathcal F:=X_0^S$ be the finite set of terminal acts. Let $\Theta$ be a finite set of private types. Type $\theta\in\Theta$ has a constraint preference represented by $(u_\theta,p_\theta,P_\theta)$ and evaluates a lottery $\lambda\in\Delta(\mathcal F)$ by $V_\theta(\lambda):=\min_{q\in P_\theta}\sum_{f\in\mathcal F}\lambda(f)\E_q[u_\theta(f)]$. We assume that every $P_\theta\subseteq\Delta^\circ$ is nonempty, compact, and convex, and maintain the common-worst-act condition imposed by \citet{strack25}. These preferences satisfy the conditions required for their concave-preference results. A maxmin representation is concave as the minimum of affine functions \citep[Lemma C.11]{strack25}, and continuity follows from finiteness of $\mathcal F$ and compactness of $P_\theta$. State independence follows directly from the representation. A singleton $P_\theta$ therefore gives a state-independent EU preference.

Let $\boldsymbol\lambda=(\lambda^\vartheta)_{\vartheta\in\Theta}$ denote an allocation profile and let $\mathrm{IMP}$ be the implementable set in \citet[][Section 2.4]{strack25}. For each type $\theta$, let $\mathrm{IC}_\theta$ contain the profiles satisfying $V_\theta(\lambda^\theta)\geq V_\theta(\lambda^\vartheta)$ for every $\vartheta\in\Theta$. For a concave type, define $b_\theta(\lambda):=\sum_{f\in\mathcal F}\lambda(f)V_\theta(\delta_f)$ and let $\mathrm B_\theta$ contain the profiles satisfying $V_\theta(\lambda^\theta)\geq\max_{\vartheta\in\Theta}b_\theta(\lambda^\vartheta)$. \citet[][Proposition 3]{strack25} then gives $\mathrm{IMP}=\bigcap_{\theta\in\Theta}\mathrm{IC}_\theta$ when every type has EU preferences, while $\mathrm{IMP}\subseteq\bigcap_{\theta\in\Theta}\mathrm B_\theta$ when every type has a concave representation. Thus, if every private type has a singleton constraint set, then every type is EU
 and the revelation principle follows their Proposition 3.
The additional screening margin studied below therefore arises only when at least one private
type retains a non-singleton constraint set. To characterize that margin when a constraint set remains non-singleton, fix one type and suppress its index. For every terminal act $f$, let $W_P(f):=\arg\min_{q\in P}\E_q[u(f)]$ denote its set of least-favorable beliefs. Define $b_P(\lambda):=\sum_{f\in\mathcal F}\lambda(f)V_P(\delta_f)$ and the \emph{screening wedge} $\Gamma_P(\lambda):=V_P(\lambda)-b_P(\lambda)$.

\begin{lemma}\label{prop:md_wedge}
For every $\lambda\in\Delta(\mathcal F)$,
$\Gamma_P(\lambda)\geq0$, with equality if and only if
$$
        \bigcap_{f\in\supp(\lambda)}W_P(f)
        \neq\varnothing.
$$
\end{lemma}

In a static direct mechanism, the value relevant for imitating $\lambda$ is $V_P(\lambda)$. The concavity bound of \citet{strack25} instead requires at least $b_P(\lambda)$. Therefore, $\Gamma_P(\lambda)$ is the gap between the static imitation value and the universal dynamic lower bound. Lemma \ref{prop:md_wedge} shows that this gap is zero precisely when all deterministic realizations in the support share a least-favorable belief. 

\subsubsection{Dual entropy and information rents}\label{sec:md_entropy}

The geometry of the dual entropy makes the least-favorable beliefs more explicit. To see this, fix $p\in\Delta^\circ$ and $P=\hat B^{\text{\normalfont\tiny KL}}_\eta(p)$. Let $\eta^\partial(p):=\min_{s\in S}\{-\log(1-p(s))\}$ be the smallest relative-entropy distance from $p$ to the boundary of the simplex $\Delta$.

\begin{lemma}[Entropy tilt]\label{lem:md_entropy_tilt}
Suppose $0<\eta<\eta^\partial(p)$. Every act $f$, for which $u\circ f$ is nonconstant, has a unique least-favorable belief $q_f^*$, and there is a unique $\tau_f>0$ such that
\begin{equation}\label{eq:md_tilt}
        q_f^*(s)
        =
        \frac{p(s)e^{-u(f(s))/\tau_f}}
        {\sum_{t\in S}p(t)e^{-u(f(t))/\tau_f}}
        \qquad\forall s\in S.
\end{equation}
\end{lemma}
Lemma \ref{lem:md_entropy_tilt} shows that  least-favorable beliefs assign more weight on states that yield lower utility. Corollary \ref{cor:md_entropy_wedge} leverages this to characterize the implementation boundary, which demonstrates that dynamic mechanisms can exploit these beliefs to lower the designer's cost.
\begin{corollary}\label{cor:md_entropy_wedge}
Suppose that, for every $f\in\supp(\lambda)$, $u\circ f$ is nonconstant. Then, $\Gamma_P(\lambda)=0$ if and only if their utility profiles are positive affine transformations of one another.
\end{corollary}

The geometry of the dual entropy therefore gives a sharp criterion for when the static imitation value and the dynamic lower bound coincide.  When the utility profiles are not positive affine transformations of one another,
their least-favorable beliefs differ and $\Gamma_P(\lambda)>0$, so the static
imitation value strictly exceeds the dynamic lower bound; the binary-state
environment below shows that this gap can be fully exploited by a dynamic
mechanism.

 Let $S=\{1,2\}$ and normalize two consequences $x^0,x^1\in X$ so that $u(x^0)=0$ and $u(x^1)=1$. Consider the three acts $f^0=(x^0,x^0)$, $f^1=(x^1,x^0)$, and $f^2=(x^0,x^1)$. There are two private types. Type $E$ has the precise belief $p=(1/2,1/2)$. Type $A$ has the same reference belief and the dual entropy constraint set $P_\eta=\hat B^{\text{\normalfont\tiny KL}}_\eta(p)$, where $0<\eta<\log2$. The constant act $f^0$ is the common worst act required by the mechanism-design framework. There is a unique $w_\eta\in(0,1)$ satisfying
\begin{equation}\label{eq:md_width}
        \eta
        =
        \frac{1+w_\eta}{2}\log(1+w_\eta)
        +
        \frac{1-w_\eta}{2}\log(1-w_\eta).
\end{equation}
Let $q_\eta:=(1-w_\eta)/2$. Then, $P_\eta=[q_\eta,1-q_\eta]$ when beliefs are identified with the probability of state $1$. For a lottery $\lambda=(\lambda_0,\lambda_1,\lambda_2)$ over the three acts, type $E$ has value $V_E(\lambda)=(\lambda_1+\lambda_2)/2$, whereas type $A$ has value
\begin{equation}\label{eq:md_binary_value}
        V_A^\eta(\lambda)
        =
        q_\eta(\lambda_1+\lambda_2)
        +
        w_\eta\min\{\lambda_1,\lambda_2\}.
\end{equation}
The formula follows by minimizing the expected utility of $(\lambda_1,\lambda_2)$ over the interval $[q_\eta,1-q_\eta]$.

Notice that, after positive rescaling, these preferences coincide with those in \citet[][Example 2]{strack25}: type $E$ has value proportional to $\lambda_1+\lambda_2$, while type $A$ has value $\alpha(\lambda_1+\lambda_2)+\min\{\lambda_1,\lambda_2\}$ with $\alpha_\eta=q_\eta/w_\eta$. The entropy radius therefore gives their curvature parameter an information-capacity interpretation. Now, let $\mathrm{IMP}_\eta$ denote the implementable set in this environment and write $\lambda^E$ and $\lambda^A$ for the lotteries assigned to the two types.

\begin{proposition}[Screening frontier]\label{prop:md_exact}
An allocation $(\lambda^E,\lambda^A)$ belongs to $\mathrm{IMP}_\eta$ if and only if
\begin{equation}\label{eq:md_EIC}
        V_E(\lambda^E)
        \geq
        V_E(\lambda^A)
\end{equation}
and
\begin{equation}\label{eq:md_AIC}
        V_A^\eta(\lambda^A)
        \geq
        q_\eta(\lambda_1^E+\lambda_2^E).
\end{equation}
\end{proposition}

Type $E$'s incentive constraint is unchanged. A static direct mechanism would additionally require $V_A^\eta(\lambda^A)\geq V_A^\eta(\lambda^E)$, whereas Proposition \ref{prop:md_exact} shows that dynamics attain the lower bound for type $A$. The exact reduction in the required utility is therefore
\begin{equation}\label{eq:md_rent_reduction}
        \Delta_\eta(\lambda^E)
        :=
        w_\eta\min\{\lambda_1^E,\lambda_2^E\}
        =
        \Gamma_{P_\eta}(\lambda^E).
\end{equation}
The gain comes from the hedge in the reduced lottery $\lambda^E$. Static imitation lets type $A$ evaluate the random allocation after its state exposures have been combined, while dynamic implementation can exploit the lower values assigned to its deterministic realizations separately.  \eqref{eq:md_rent_reduction} shows that the resulting screening gain is precisely the wedge in Lemma \ref{prop:md_wedge}.

\begin{corollary}\label{cor:md_capacity}
If $\lambda_1^E\lambda_2^E>0$, then $\Delta_\eta(\lambda^E)$ is strictly increasing in $\eta$. Moreover,
\begin{equation}\label{eq:md_small_eta}
        \Delta_\eta(\lambda^E)
        =
        \sqrt{2\eta}\,
        \min\{\lambda_1^E,\lambda_2^E\}
        +
        O(\eta^{3/2})
        \qquad\text{as }\eta\downarrow0.
\end{equation}
\end{corollary}

As $\eta\downarrow0$, the constraint set collapses to the reference belief and the screening gain vanishes. As $\eta\uparrow\log2$, $w_\eta\uparrow1$ and the gain converges to $\min\{\lambda_1^E,\lambda_2^E\}$. For fixed total exposure $t=\lambda_1^E+\lambda_2^E$, the gain is largest at $\lambda_1^E=\lambda_2^E=t/2$. Information capacity therefore has a precise rent implication: a larger admissible distortion of beliefs increases the value of dynamic screening, especially for allocations that hedge opposing state exposures.

\section{Conclusion}\label{sec:conclusion}

This paper studies the behavioral and aggregate implications of constraint preferences by separating two questions that are often treated jointly: which individual belief restrictions a collective decision should respect, and where those restrictions come from. The first main result links a subclass of constraint preferences to rationally inattentive behavior by showing that posterior unit cost coincides with the divergence defining the constraint set. The second main result introduces a Pareto principle for maxmin preferences and shows that collective beliefs must satisfy the constraints of all relevant individuals. Together, these results produce a clear division of labor: the aggregation principle determines how individual constraints are combined, while information-processing technologies determine their structure or geometry.

The applications illustrate different consequences of this distinction. In financial markets, shared beliefs in the intersection are capitalized into equilibrium prices, and using an aggregation rule inconsistent with the underlying information technology can generate financial losses. In prediction markets, the geometry of individual belief sets determines when inherited prices become self-confirming, particularly after rare signals. In mechanism design, the same geometry determines least-favorable beliefs and the value of dynamic screening. Information technologies therefore matter not only for individual behavior, but also for collective beliefs, market prices, information aggregation, and information rents.
\appendix

\section{Appendix: Regularity Conditions}\label{app:conditions}

We impose some regularity conditions on Bregman balls that are standard in information geometry, where, for example, all balls are typically assumed to have the same radius \citep[e.g.,][]{small18,topo18}. Throughout, $S$ is finite.
We maintain the assumption that the intersection of individual belief sets is nonempty, and all reference beliefs satisfy
$p_i\in\Delta^{\circ}$ for every $i\in I$. Whenever a result assumes that an intersection of belief sets is a singleton, the singleton element is assumed to lie in $\Delta^{\circ}$.  This guarantees that the derivatives used in the singleton-intersection arguments are well-defined. This assumption is mild because for linear and power pooling, for example, interior solutions hold under the ``coherence'' assumption $\bigcup_{i\in I}\text{supp}(p_i)=S$, which is standard in the belief-aggregation literature \citep[e.g.,][]{diet21,agg24}. For geometric and multiplicative pooling, interiority holds under the relatively stronger assumption $\bigcap_{i\in I}\text{supp}(p_i)=S$,  which requires  every $p_i$ to have full support.

\par Notice that the entropy generator $G(q)=\sum_{s\in S}q(s)\log q(s)$ is not differentiable on the boundary of the simplex. Thus, whenever a generator $G$ is defined only on an open convex set $\mathcal D$ containing $\Delta^{\circ}$, the associated Bregman divergence is understood on $\Delta$ through its lower-semicontinuous extension whenever this extension exists. For the entropy generator, we use the standard conventions $0\log 0=0$ and $a\log(a/0)=+\infty$ for $a>0$.

For primal Bregman-ball results, the maintained conditions are the following. The generator $G$ is twice continuously differentiable and convex on an open convex domain $\mathcal D\subseteq\mathbb R^S$ containing $\Delta^{\circ}$.\footnote{These conditions on $G$ are standard in economics \citep[e.g.,][Assumption 1]{HebertWoodford2021}. } For every $i\in I$, define the extended-valued function
$$
        F_i(q):=
        \begin{cases}
        D_G(p_i\Vert q), & q\in\Delta\cap\mathcal D,\\
        +\infty, & q\in\Delta\setminus\mathcal D.
        \end{cases}
$$
We assume that $F_i$ is convex and lower semicontinuous on $\Delta$. We also assume that, for every $q\in\Delta\cap\mathcal D$, the Hessian $\nabla^2G(q)$ is positive definite on the tangent space
$T_\Delta:=\big\{z\in\mathbb R^S:\sum_{s\in S}z(s)=0\big\}.$
These conditions imply that the primal balls
$B^G_{\eta_i}(p_i)=\{q\in\Delta:F_i(q)\leq \eta_i\}$
are closed and convex subsets of $\Delta$, which are the two required properties of belief sets in maxmin preferences. These assumptions are necessary because \citet{topo18} show that, while dual balls are always convex, primal balls need not be, so we must restrict attention to primal balls that are convex. For example, they include primal entropy balls because, for $p_i\in\Delta^{\circ}$, $D_{\mathrm{KL}}(p_i\Vert q)$ is convex and lower semicontinuous on $\Delta$, finite on $\Delta^{\circ}$, and the entropy Hessian is positive definite on $T_\Delta$.

For dual Bregman-ball results, no constraint qualification is imposed. We use only the following implication. If $q^*\in\Delta^{\circ}$ and
$\bigcap_{i\in I}\big\{q\in\Delta:f_i(q)\leq0\big\}=\{q^*\},
$ then there exist multipliers
$(\lambda_i)_{i\in A(q^*)}$, with $\lambda_i\geq0$ and not all zero, and a scalar $\zeta\in\mathbb R$, such that
$\sum_{i\in A(q^*)}\lambda_i\nabla f_i(q^*)=\zeta\mathbf 1 $, where
$A(q^*)=\{i\in I:f_i(q^*)=0\}$ and each $f_i$ is convex on $\Delta$ and differentiable in a neighborhood of $q^*$.
This implication is proved in Appendix \ref{app:proofs} as Lemma \ref{lem:active_normal_auto}. Applied to
$f_i(q)=D_G(q\Vert p_i)-\eta_i$, for example, it gives the active-normal relation used in the proofs of Proposition \ref{thm:geo1} and Proposition \ref{prop:separable_pool}.

For Theorem \ref{thm:strict_ri}, $G$ is convex and differentiable at the full-support reference belief $p$, and $D_G(q\lVert p)$ is well defined for every posterior under consideration. For the reverse-order results in Appendix \ref{sec:micro_primal}, $G$ is continuously differentiable in a neighborhood of every posterior at which $D_G(p\lVert q)$ is evaluated.

\section{Appendix: Proofs}\label{app:proofs}

\subsection{Proofs from Section \ref{sec:strict_ri}}

\begin{proof}[Proof of Lemma \ref{thm:dilution}]
Fix $(Y,\sigma)$, $p\in\Delta^\circ$, and $\alpha\in[0,1]$. For every $y\in Y$,
$\pi_y(\sigma^\alpha,p)
        =
        \sum_{s\in S}p(s)\alpha\sigma_s(y)
        =
        \alpha\pi_y(\sigma,p).$
If $\alpha>0$ and $\pi_y(\sigma,p)>0$, Bayes' rule gives
$q_y(\sigma^\alpha,p)(s)
        =
        \frac{p(s)\alpha\sigma_s(y)}
        {\alpha\pi_y(\sigma,p)}
        =
        q_y(\sigma,p)(s).$
The null signal has probability $1-\alpha$ and, whenever $\alpha<1$, posterior $p$. Hence, uniform posterior separability and $D_G(p\lVert p)=0$ imply
\begin{align*}
        C(\sigma^\alpha,p;Y\cup\{o\})
=
        \sum_{y\in Y}
        \alpha\pi_y(\sigma,p)
        D_G(q_y(\sigma,p)\lVert p)
        +(1-\alpha)D_G(p\lVert p)=
        \alpha C(\sigma,p;Y).
\end{align*}
The boundary cases $\alpha=0$ and $\alpha=1$ are included in the same calculation, with zero-probability signals contributing zero.
\end{proof}

\begin{proof}[Proof of Observation \ref{prop:degeneracy}]
By hypothesis, there exist $(Y,\sigma)$ and $y^\star$ with
$\pi_{y^\star}(\sigma,p)>0$
and
$q_{y^\star}(\sigma,p)=q$.
For any $\alpha\in(0,1]$, dilution leaves the posterior after $y^\star$ unchanged and Lemma \ref{thm:dilution} gives
$C(\sigma^\alpha,p;Y\cup\{o\})
        =
        \alpha C(\sigma,p;Y).$
If the original cost is zero, take $\alpha=1$. Otherwise choose
$\alpha
        :=
        \min\big\{
        1,\frac{\eta}{C(\sigma,p;Y)}
        \big\}.$
Then, the diluted experiment still induces $q$ after $y^\star$ and has expected cost at most $\eta$.
\end{proof}

\begin{proof}[Proof of Observation \ref{prop:unitprice}]
For part 1, fix $\varepsilon\in(0,1]$ and any admissible $(Y,\sigma,y^\star)$ in the definition of $c^\star(\varepsilon)$. Since
$\pi_{y^\star}(\sigma,p)\geq\varepsilon$,
$\frac{C(\sigma,p;Y)}{\varepsilon}
        \geq
        \frac{C(\sigma,p;Y)}
        {\pi_{y^\star}(\sigma,p)}
        \geq
        \kappa(q|p).$
Multiplying by $\varepsilon$ and taking the infimum yields
$c^\star(\varepsilon)\geq\varepsilon\kappa(q|p)$. For part 2, fix $\delta>0$. By the definition of $\kappa(q|p)$ as an infimum, there is an admissible $(\bar Y,\bar\sigma,\bar y^\star)$ such that, writing
$\bar\pi:=\pi_{\bar y^\star}(\bar\sigma,p)>0$,
$\frac{C(\bar\sigma,p;\bar Y)}{\bar\pi}
        \leq
        \kappa(q|p)+\delta.$
Set $\bar\varepsilon:=\bar\pi$. For any
$\varepsilon\in(0,\bar\varepsilon]$,
let $\alpha:=\varepsilon/\bar\varepsilon$. In the diluted experiment, $\bar y^\star$ occurs with probability $\varepsilon$, still induces $q$, and by Lemma \ref{thm:dilution} has total cost
$\alpha C(\bar\sigma,p;\bar Y)$. Therefore,
$c^\star(\varepsilon)
        \leq
        \alpha C(\bar\sigma,p;\bar Y)
        =
        \varepsilon
        \frac{C(\bar\sigma,p;\bar Y)}{\bar\pi}
        \leq
        \varepsilon(\kappa(q|p)+\delta).$
\end{proof}

\begin{lemma}\label{lem:o}
Let $G$ be differentiable at $p$. If $h_k\to0$ and
$D_G(p+h_k\lVert p)$
is well defined for every $k$, then
$D_G(p+h_k\lVert p)
        =
        o(\|h_k\|).$
\end{lemma}

\begin{proof}[Proof of Lemma \ref{lem:o}]
Differentiability gives
$G(p+h_k)=G(p)+\langle\nabla G(p),h_k\rangle+o(\|h_k\|)$.
Substituting into the definition of the Bregman divergence yields
$D_G(p+h_k\lVert p)=o(\|h_k\|)$.
\end{proof}

\begin{proof}[Proof of Theorem \ref{thm:strict_ri}]
Fix $q\in\Delta$. For the lower bound, take any admissible $(Y,\sigma,y^\star)$ in Definition \ref{def:unitcost}. Uniform posterior separability and nonnegativity of the Bregman divergence imply
$C(\sigma,p;Y)
        =
        \sum_{y\in Y}
        \pi_y(\sigma,p)
        D_G(q_y(\sigma,p)\lVert p)\geq
        \pi_{y^\star}(\sigma,p)
        D_G(q_{y^\star}(\sigma,p)\lVert p)
        =
        \pi_{y^\star}(\sigma,p)
        D_G(q\lVert p).$
Dividing by the positive probability of $y^\star$ and taking the infimum gives
$\kappa(q|p)
        \geq
        D_G(q\lVert p).$

For the reverse inequality, choose $\varepsilon\in(0,1)$ such that
$\varepsilon
        <
        \min\big\{
        1,
        \min_{s:q(s)>0}
        \frac{p(s)}{q(s)}
        \big\}.$
Let $Y=\{y^\star,y^0\}$ and define
$\sigma_s(y^\star)
        :=
        \varepsilon\frac{q(s)}{p(s)}$ and $\sigma_s(y^0)
        :=
        1-\sigma_s(y^\star).$
These are valid probabilities because $p$ has full support and $\varepsilon$ satisfies the displayed inequality. Moreover,
$\pi_{y^\star}(\sigma,p)
        =
        \varepsilon$ and $
        q_{y^\star}(\sigma,p)
        =
        q.$
The posterior after $y^0$ is
$r_\varepsilon(s)
        :=
        \frac{p(s)-\varepsilon q(s)}{1-\varepsilon}
        \in\Delta^\circ.$
Hence, $C(\sigma,p;Y)
        =
        \varepsilon D_G(q\lVert p)
        +(1-\varepsilon)
        D_G(r_\varepsilon\lVert p),$
and therefore
$\frac{C(\sigma,p;Y)}{\varepsilon}
        =
        D_G(q\lVert p)
        +
        \frac{1-\varepsilon}{\varepsilon}
        D_G(r_\varepsilon\lVert p).$
Now,
$r_\varepsilon-p
        =
        \frac{\varepsilon}{1-\varepsilon}(p-q),$
so $\|r_\varepsilon-p\|=O(\varepsilon)$. Lemma \ref{lem:o} implies
$D_G(r_\varepsilon\lVert p)=o(\varepsilon)$.
Taking $\varepsilon\downarrow0$ therefore gives
$\kappa(q|p)\leq D_G(q\lVert p)$.
Together with the previous inequality, $\kappa(q|p)=D_G(q\lVert p)$.

Finally, by Definition \ref{def:strictRI},
$(u,p,P)$ is strictly rationally inattentive under $(C,\eta)$ if and only if
$P
        =
        \{q:\kappa(q|p)\leq\eta\}.$
Using the identity just proved, the right-hand side equals
$\{q:D_G(q\lVert p)\leq\eta\}=\hat B^G_\eta(p)$.
This proves the equivalence.
\end{proof}

\subsection{Proofs from Section \ref{sec:aggregation}}

\begin{lemma}\label{lem:min_directional}
Let $P\subseteq\Delta$ be nonempty and compact and define
$\Phi(z):=\min_{p\in P}p\cdot z$
for $z\in\R^S$. Then, for every $z,v\in\R^S$, the right directional derivative exists and equals
$$
        \Phi'(z;v)
        =
        \min_{p\in M(z)}p\cdot v,
        \qquad
        M(z):=\arg\min_{p\in P}p\cdot z.
$$
\end{lemma}

\begin{proof}[Proof of Lemma \ref{lem:min_directional}]
Let $p^*\in M(z)$ minimize $p\cdot v$ over the compact set $M(z)$. For every $t>0$,
$\Phi(z+tv)
        \leq
        p^*\cdot(z+tv)
        =
        \Phi(z)+t p^*\cdot v,$
so
$\limsup_{t\downarrow0}
        \frac{\Phi(z+tv)-\Phi(z)}{t}
        \leq
        p^*\cdot v.$

For the reverse bound, take a sequence $t_k\downarrow0$ along which the difference quotient
converges to its lower limit and choose
$p_k\in\arg\min_{p\in P}p\cdot(z+t_kv)$.
Compactness gives a convergent subsequence, still denoted $p_k$, with $p_k\to\bar p\in P$. Optimality implies
$p_k\cdot(z+t_kv)
        \leq
        p\cdot(z+t_kv)$ $\forall p\in P.$
Letting $k\to\infty$ gives
$\bar p\cdot z\leq p\cdot z$ for every $p$, so
$\bar p\in M(z)$. Moreover,
$\frac{\Phi(z+t_kv)-\Phi(z)}{t_k}
=
\frac{p_k\cdot z-\Phi(z)}{t_k}
+
p_k\cdot v
\geq
p_k\cdot v.$
Taking limits gives a lower bound
$\bar p\cdot v\geq\min_{p\in M(z)}p\cdot v=p^*\cdot v$.
The upper and lower bounds coincide.
\end{proof}

\begin{proof}[Proof of Theorem \ref{thm:agg}]
\emph{Sufficiency.}
Suppose \eqref{eq:agg} holds and
$f\succsim_i^U g$ for every $i$. Fix $p\in P_0$. For every $i$ with $\varphi_i>0$, the inclusion in \eqref{eq:agg} implies $p\in P_i$. Hence, $\E_p[u_i(f)]
        \geq
        \E_p[u_i(g)].$
Therefore, $\E_p[u_0(f)]-\E_p[u_0(g)]=
        \sum_{i\in I}
        \varphi_i
        \left(
        \E_p[u_i(f)]-\E_p[u_i(g)]
        \right)
        \geq0.$
This holds for every $p\in P_0$, so
$\min_{p\in P_0}\E_p[u_0(f)]
        \geq
        \min_{p\in P_0}\E_p[u_0(g)].$
Thus, $f\succsim_0 g$, which proves unambiguous maxmin Pareto dominance.

\medskip\noindent
\emph{Necessity: utility aggregation.}
Restrict the Pareto principle to constant acts. If
$x\succsim_i y$ for every $i$, then beliefs play no role and
$x\succsim_i^U y$ for every $i$. Thus,
$x\succsim_0 y$.

Define
$U:X\to\R^n$ and $U(x)
        :=
        (u_1(x),\dots,u_n(x)).$
If $U(x)=U(y)$, then $x\sim_i y$ for every individual, so the Pareto implication applies in both directions and gives $u_0(x)=u_0(y)$. Hence, there exists a well-defined function
$\upsilon:U(X)\to\R$
such that
$\upsilon(U(x))
        =
        u_0(x).$
Since all utility functions are affine, $\upsilon$ is affine on the convex set $U(X)$.

By c-diversity, for every $i$, there are $x_i,y_i\in X$ such that
$x_i\succ_i y_i$
and
$x_i\sim_j y_i$ for all $j\neq i$. Hence, $U(x_i)-U(y_i)
        =
        a_i e_i$
for some $a_i>0$, where $e_i$ is the $i$th coordinate vector. These differences span $\R^n$, so the affine hull of $U(X)$ is $\R^n$. Therefore, $\upsilon$ has a unique affine extension to $\R^n$, and there are $\varphi\in\R^n$ and $\gamma\in\R$ such that
$u_0
        =
        \sum_{i\in I}\varphi_i u_i+\gamma.$

For the c-diverse pair $(x_i,y_i)$, the Pareto implication gives
$u_0(x_i)\geq u_0(y_i)$. Thus,
$\varphi_i
        \big(
        u_i(x_i)-u_i(y_i)
        \big)
        \geq0.$
Since the term in parentheses is strictly positive,
$\varphi_i\geq0$ for every $i$. Since $u_0$ is nonconstant, $\varphi\neq0$.

\medskip\noindent
\emph{Necessity: belief inclusion.}
Fix $k$ with $\varphi_k>0$. Suppose, toward a contradiction, that
$P_0\nsubseteq P_k$. Choose
$\bar p\in P_0\setminus P_k$. By strict separation of the point $\bar p$ from the compact convex set $P_k$ (strict hyperplane separation theorem), there exist $a\in\R^S$ and $c\in\R$ such that
\begin{equation}\label{eq:sep}
        \bar p\cdot a
        <
        c
        \leq
        p\cdot a
        \qquad\forall p\in P_k.
\end{equation}
Multiplying both $a$ and $c$ by the same positive scalar if necessary, assume that $|a(s)-c|\leq 1$ $\forall s$. Let $m:=
        \min_{p\in P_0}p\cdot a$, $F_a:=
        \arg\min_{p\in P_0}p\cdot a,$
and
$d:=a-c\mathbf1$. Then, \eqref{eq:sep} implies
\begin{equation}\label{eq:pdpositive}
        p\cdot d\geq0
        \qquad\forall p\in P_k,
\end{equation}
while for every $q\in F_a$,
\begin{equation}\label{eq:qdnegative}
        q\cdot d
        =
        m-c
        <
        0.
\end{equation}
 
By c-diversity, choose $x,y\in X$ with
$u_k(x)>u_k(y)$ and
$u_j(x)=u_j(y)$ $\forall j\neq k$. Define
$$
        f^*(s)
        :=
        \frac{1+d(s)}2x
        +
        \frac{1-d(s)}2y,
        \qquad
        g^*(s)
        :=
        \frac{1-d(s)}2x
        +
        \frac{1+d(s)}2y.
$$
Let $\rho:=u_k(x)-u_k(y)>0$. Affinity gives
$u_k(f^*(s))-u_k(g^*(s))
        =
        \rho d(s),$
while
$u_j(f^*(s))-u_j(g^*(s))=0$ for $j\neq k$. Therefore, \eqref{eq:pdpositive} implies
$f^*
        \succsim_i^U
        g^*$ $\forall i\in I.$

We will now construct a common base act that makes $F_a$ the planner's set of minimizing priors. Since $u_0$ is nonconstant and affine on the convex set $X$, $u_0(X)$ contains a nondegenerate interval. Since $S$ is finite, choose $A\in\R$ and $B>0$ sufficiently small that
$A+Ba(s)\in u_0(X)$ for every state. Choose $h(s)\in X$ satisfying
\begin{equation}\label{eq:baseact}
        u_0(h(s))
        =
        A+Ba(s).
\end{equation}
For $t\in(0,1)$, define
$f_t:=(1-t)h+t f^*$ and $
        g_t:=(1-t)h+t g^*.$
Affinity and the preceding robust comparison imply
$f_t\succsim_i^U g_t$ for every $i$ and every $t>0$.

Let
$z(s):=u_0(h(s))$,
$v_f(s):=u_0(f^*(s))$,
and
$v_g(s):=u_0(g^*(s))$.
By \eqref{eq:baseact}, the minimizers of $p\cdot z$ over $P_0$ are exactly $F_a$. Moreover, using the already-established utility aggregation,
$v_f(s)-v_g(s)
        =
        \varphi_k\rho d(s).$
Hence, by \eqref{eq:qdnegative},
$q\cdot(v_f-v_g)
        =
        \varphi_k\rho(m-c)
        <
        0$ $\forall q\in F_a.$

Define
$\Phi(w):=\min_{p\in P_0}p\cdot w$.
Then, $V_0(f_t)
        =
        \Phi(z+t(v_f-z))$ and $V_0(g_t)
        =
        \Phi(z+t(v_g-z)).$
Lemma \ref{lem:min_directional} implies that the right derivative at $t=0$ of
$V_0(f_t)-V_0(g_t)$ equals
\begin{align*}
\min_{q\in F_a}q\cdot(v_f-z)
-
\min_{q\in F_a}q\cdot(v_g-z)
=
\varphi_k\rho(m-c)
<
0,
\end{align*}
where the equality uses the fact that
$q\cdot(v_f-v_g)$
is constant on $F_a$. Since
$V_0(f_0)=V_0(g_0)$,
it follows that
$V_0(f_t)<V_0(g_t)$ for all sufficiently small $t>0$. This contradicts unambiguous maxmin Pareto dominance. Thus, $P_0\subseteq P_k$.

Since $k$ was arbitrary among the indices satisfying $\varphi_k>0$,
$P_0
        \subseteq
        \bigcap_{i:\varphi_i>0}P_i.$
\end{proof}

\begin{proof}[Proof of Observation \ref{prop:pareto_comparison}]
Suppose first that $\succsim_0$ satisfies unambiguous maxmin Pareto dominance. By Theorem \ref{thm:agg}, there exist a nonzero vector $\varphi\in\R^n_+$ and $\gamma\in\R$ such that
$u_0
        =
        \sum_{i\in I}\varphi_i u_i+\gamma$ and
        $P_0
        \subseteq
        \bigcap_{i:\varphi_i>0}P_i.$
Take any $f,g\in F$ such that
$f\succsim_i^U g$
for every $i\in I$.
Fix any $q\in P_0$. For every $i$ with $\varphi_i>0$, the intersection restriction implies $q\in P_i$, and hence
$\E_q[u_i(f)]
        \geq
        \E_q[u_i(g)]$, so $\E_q[u_0(f)]-\E_q[u_0(g)]=
        \sum_{i\in I}
        \varphi_i
        \left(
        \E_q[u_i(f)]-\E_q[u_i(g)]
        \right)\geq0.$
Since $q\in P_0$ was arbitrary,
$f\succsim_0^U g$.
Thus, revealed unambiguous Pareto holds. Conversely, suppose revealed unambiguous Pareto holds. If
$f\succsim_i^U g$  $\forall i\in I$, then
$f\succsim_0^U g$, so
$\E_q[u_0(f)]
        \geq
        \E_q[u_0(g)]$ $\forall q\in P_0.$
Taking minima over $P_0$ gives
$\min_{q\in P_0}\E_q[u_0(f)]
        \geq
        \min_{q\in P_0}\E_q[u_0(g)],$
which is
$f\succsim_0 g$.
Hence, unambiguous maxmin Pareto dominance holds.
\end{proof}

\subsection{Proofs from Section \ref{sec:bregman}}

\begin{proof}[Proof of Proposition \ref{thm:pos}]
Let $\mathcal B:=\bigcap_{i\in I}B^G_{\eta_i}(p_i)$. By assumption, $\mathcal B\neq\varnothing$. For $q\in\Delta$, define $\Phi(q):=\max_{i\in I}\{F_i(q)-\eta_i\}$, where $F_i$ is the extended-valued function from Appendix \ref{app:conditions}. Since each $F_i$ is convex and lower semicontinuous on $\Delta$, the function $\Phi$ is convex and lower semicontinuous on $\Delta$. Since the state space $S$ is finite, $\Delta$ is compact, so $\Phi$ attains a minimum on $\Delta$. Let $q^\star\in\Delta$ be a minimizer. Choose $\bar q\in\mathcal B$. Then, $\Phi(\bar q)\leq0$, so $\Phi(q^\star)\leq0$. Hence, $q^\star\in\mathcal B$, which implies $F_i(q^\star)<+\infty$ for every $i\in I$, and therefore $q^\star\in\mathcal D$.

Let $A(q^\star):=\{i\in I:\Phi(q^\star)=F_i(q^\star)-\eta_i\}$. Since $q^\star\in\mathcal D$, each active $F_i$ is differentiable at $q^\star$. By the finite-dimensional max rule for the subdifferential of a finite maximum of differentiable convex functions, $\partial\Phi(q^\star)=\operatorname{co}\{\nabla_qD_G(p_i\Vert q^\star):i\in A(q^\star)\}$ \citep[][Exercise 4.1.44.(b)]{convex10}. Since $q^\star$ minimizes $\Phi$ over $\Delta$, the convex first-order condition gives $0\in\partial\Phi(q^\star)+N_\Delta(q^\star)$, where $N_\Delta(q^\star):=\{v\in\mathbb R^S:\langle v,z-q^\star\rangle\leq0\hspace{0.07in} \forall z\in\Delta\}$ \citep[][Exercise 4.1.49]{convex10}. Hence, there exist weights $(\lambda_i)_{i\in A(q^\star)}$, with $\lambda_i\geq0$ and $\sum_{i\in A(q^\star)}\lambda_i=1$, and a vector $v\in N_\Delta(q^\star)$ such that $\sum_{i\in A(q^\star)}\lambda_i\nabla_qD_G(p_i\Vert q^\star)+v=0$.

For each $i\in A(q^\star)$, differentiating $D_G(p_i\Vert q)=G(p_i)-G(q)-\langle\nabla G(q),p_i-q\rangle$ with respect to $q$ gives $\nabla_qD_G(p_i\Vert q^\star)=\nabla^2G(q^\star)(q^\star-p_i)$. Let $\bar p:=\sum_{i\in A(q^\star)}\lambda_i p_i$. Since each $p_i\in\Delta$ and $\Delta$ is convex, $\bar p\in\Delta$. Substituting the derivative formula into the first-order condition yields $v=\nabla^2G(q^\star)(\bar p-q^\star)$. Since $v\in N_\Delta(q^\star)$ and $\bar p\in\Delta$, $\langle v,\bar p-q^\star\rangle\leq0$. Therefore, $\langle\nabla^2G(q^\star)(\bar p-q^\star),\bar p-q^\star\rangle\leq0$. However, $\bar p-q^\star\in T_\Delta:=\left\{z\in\mathbb R^S:\sum_{s\in S}z(s)=0\right\}$, and Appendix \ref{app:conditions} assumes that $\nabla^2G(q^\star)$ is positive definite on $T_\Delta$. Hence, $\bar p-q^\star=0$. Thus, $q^\star=\bar p\in\operatorname{co}\{p_i:i\in A(q^\star)\}\subseteq\operatorname{co}(p_i)_{i\in I}$. Since $q^\star\in\mathcal B$, then $\bigcap_{i\in I}B^G_{\eta_i}(p_i)\cap\operatorname{co}(p_i)_{i\in I}\neq\varnothing$.
\end{proof}

\begin{proof}[Proof of Corollary \ref{thm:neutral}]
Let
$\mathcal{C}:=\text{co}(p_{i})_{i\in I}.$
Since $\mathcal{C}$ is the convex hull of finitely many points in the finite-dimensional simplex
$\Delta$, it is compact. For each $i\in I$, the map
$q\mapsto D_G(p_i\Vert q)$
is continuous on $\mathcal{C}$, so it attains a finite maximum on $\mathcal{C}$. Define $\eta_i:=\max_{q\in \mathcal{C}}D_G(p_i\Vert q)$ $\forall i\in I .$
Now, take any $q\in \mathcal{C}$. By the definition of $\eta_i$,
$D_G(p_i\Vert q)\leq \eta_i$ $\forall i\in I .$
Therefore, $q\in B^G_{\eta_i}(p_i)$ $\forall i\in I$, so
$q\in \bigcap_{i\in I} B^G_{\eta_i}(p_i).$
Since $q\in \mathcal{C}$ was arbitrary, we obtain
$\text{co}(p_{i})_{i\in I}
        \subseteq
        \bigcap_{i\in I} B^G_{\eta_i}(p_i).$
\end{proof}

\begin{proof}[Proof of Corollary \ref{thm:pos1}]
This follows directly from Proposition \ref{thm:pos} by imposing a singleton.
\end{proof}

\begin{lemma}\label{lem:active_normal_auto}
Let $q^*\in\Delta^{\circ}$. For each $i\in I$, let $f_i:\Delta\to\mathbb R$ be convex and differentiable in a neighborhood of $q^*$. Suppose
$\bigcap_{i\in I}\{q\in\Delta:f_i(q)\leq0\}=\{q^*\}.$
Let $A(q^*):=\{i\in I:f_i(q^*)=0\}$. Then, $A(q^*)\neq\varnothing$, and there exist multipliers $(\lambda_i)_{i\in A(q^*)}$, with $\lambda_i\geq0$ and not all zero, and a scalar $\zeta\in\mathbb R$, such that
$\sum_{i\in A(q^*)}\lambda_i\nabla f_i(q^*)=\zeta\mathbf 1 .$
\end{lemma}

\begin{proof}[Proof of Lemma \ref{lem:active_normal_auto}]
If $A(q^*)=\varnothing$, then $f_i(q^*)<0$  $\forall i\in I$. Since $I$ is finite and every $f_i$ is continuous at $q^*$, there exists a relative neighborhood $U$ of $q^*$ in $\Delta$ such that $f_i(q)<0$  $\forall q\in U$ and $\forall i\in I$. This contradicts
$\bigcap_{i\in I}\{q\in\Delta:f_i(q)\leq0\}=\{q^*\}.$
Hence, $A(q^*)\neq\varnothing$. Recall
$T_\Delta=\left\{z\in\mathbb R^S:\sum_{s\in S}z(s)=0\right\}.$
For each $i\in A(q^*)$, view $d\mapsto \langle\nabla f_i(q^*),d\rangle$ as a linear functional on $T_\Delta$. Gordan's theorem of the alternative \citep[][Theorem 2.4.5, p. 31]{gord69}, applied to the matrix of the negative active gradients restricted to the finite-dimensional vector space $T_\Delta$, implies that exactly one of the following two alternatives holds:
$$
\begin{array}{ll}
\text{\normalfont (a)} & \text{\normalfont there exists } d\in T_\Delta \text{ \normalfont s.t. }
        \langle\nabla f_i(q^*),d\rangle<0 \text{ \normalfont for every } i\in A(q^*),\\
\text{\normalfont (b)} & \text{\normalfont there exist } (\lambda_i)_{i\in A(q^*)}, \text{\normalfont with } \lambda_i\geq0
        \text{ \normalfont and not all zero, s.t. }
        \sum_{i\in A(q^*)}\lambda_i\nabla f_i(q^*) \text{ \normalfont annihilates } T_\Delta .
\end{array}
$$
We show that alternative (a) is impossible. Suppose otherwise, and fix $d\in T_\Delta$ satisfying the inequalities in (a). Since $q^*\in\Delta^\circ$, there exists $\bar\varepsilon>0$ such that $q^*+\varepsilon d\in\Delta^\circ$ for every $\varepsilon\in(0,\bar\varepsilon)$. For every active $i\in A(q^*)$, differentiability of $f_i$ at $q^*$ gives
$f_i(q^*+\varepsilon d)
        =
        f_i(q^*)+\varepsilon\langle\nabla f_i(q^*),d\rangle+o(\varepsilon)
        <0$
for all sufficiently small $\varepsilon>0$, because $f_i(q^*)=0$ and
$\langle\nabla f_i(q^*),d\rangle<0$. For every inactive $i\in I\setminus A(q^*)$, we have $f_i(q^*)<0$, so continuity of $f_i$ at $q^*$ implies $f_i(q^*+\varepsilon d)<0$ for all sufficiently small $\varepsilon>0$. Therefore, for all sufficiently small $\varepsilon>0$,
$q^*+\varepsilon d\in \bigcap_{i\in I}\{q\in\Delta:f_i(q)\leq0\}.$
Moreover, $q^*+\varepsilon d\neq q^*$ because $d\neq0$; if $d=0$, the strict inequalities in alternative (a) could not hold. This contradicts the singleton-intersection assumption. Hence, alternative (a) is impossible, so alternative (b) must hold. Thus, there exist multipliers $(\lambda_i)_{i\in A(q^*)}$, with $\lambda_i\geq0$ and not all zero, such that
$\langle \sum_{i\in A(q^*)}\lambda_i\nabla f_i(q^*),d\rangle=0$ $\forall d\in T_\Delta .$
Equivalently, $\sum_{i\in A(q^*)}\lambda_i\nabla f_i(q^*)\in T_\Delta^\perp$. Since
$T_\Delta^\perp=\operatorname{span}\{\mathbf 1\}$, there exists $\zeta\in\mathbb R$ such that
$\sum_{i\in A(q^*)}\lambda_i\nabla f_i(q^*)=\zeta\mathbf 1 .$
\end{proof}

\begin{proof}[Proof of Proposition \ref{thm:geo1}]
For each $i\in I$, define $f_i(q):=\KL(q\Vert p_i)-\eta_i$. Since $p_i\in\Delta^{\circ}$, each $f_i$ is convex on $\Delta$ and differentiable in a neighborhood of $\hat q_0$. By assumption, $\bigcap_{i\in I}\hat B^{\text{\normalfont\tiny KL}}_{\eta_i}(p_i)=\{\hat q_0\}$, and Appendix \ref{app:conditions} assumes $\hat q_0\in\Delta^{\circ}$. Lemma \ref{lem:active_normal_auto} gives multipliers $\lambda_i\geq0$, not all zero, supported on $A:=\{i\in I:\KL(\hat q_0\Vert p_i)=\eta_i\}$, and a scalar $\zeta\in\mathbb R$, such that $\sum_{i\in A}\lambda_i\nabla_q\KL(\hat q_0\Vert p_i)=\zeta\mathbf 1$. Let $\Lambda:=\sum_{i\in A}\lambda_i>0$. Replacing each $\lambda_i$ by $\lambda_i/\Lambda$, and setting $\lambda_i=0$ for $i\notin A$, we may assume without loss that $\lambda\in\Delta(I)$.

Since all reference beliefs and $\hat q_0$ have full support, $\partial\KL(q\Vert p_i)/\partial q(s)$ evaluated at $q=\hat q_0$ equals $\log\hat q_0(s)-\log p_i(s)+1$ $\forall s\in S$. Substituting this derivative into the active-normal relation gives $\sum_{i\in I}\lambda_i[\log\hat q_0(s)-\log p_i(s)+1]=\zeta$  $\forall s\in S$. Since $\sum_{i\in I}\lambda_i=1$, this is equivalent to $\log\hat q_0(s)=\zeta-1+\sum_{i\in I}\lambda_i\log p_i(s)$ $\forall s\in S$. Exponentiating gives $\hat q_0(s)=K\prod_{i\in I}p_i(s)^{\lambda_i}$ for every $s$, where $K:=\exp(\zeta-1)>0$ is independent of $s$. Since $\hat q_0\in\Delta$, summing over states gives $K=1/\sum_{t\in S}\prod_{i\in I}p_i(t)^{\lambda_i}$. Substituting this value of $K$ gives \eqref{eq:geopool}.
\end{proof}

\begin{proof}[Proof of Proposition \ref{prop:separable_pool}]
For each $i\in I$, let $f_i(q):=D_G(q\Vert p_i)-\eta_i$. Since $G$ is convex and differentiable in a neighborhood of $\tilde q_0$, each $f_i$ is convex and differentiable in a neighborhood of $\tilde q_0$. By assumption, $\bigcap_{i\in I}\hat B^G_{\eta_i}(p_i)=\{\tilde q_0\}$ and   $\tilde q_0\in\Delta^\circ$ (Appendix \ref{app:conditions}). Lemma \ref{lem:active_normal_auto} gives multipliers $\lambda_i\geq0$, not all zero, supported on $A:=\{i\in I:D_G(\tilde q_0\Vert p_i)=\eta_i\}$, and $\zeta\in\mathbb R$, such that
$\sum_{i\in A}\lambda_i\nabla_qD_G(\tilde q_0\Vert p_i)=\zeta\mathbf 1.
$
For every $i\in A$, differentiating $D_G(q\Vert p_i)$ with respect to $q$ gives $\nabla_qD_G(q\Vert p_i)=\nabla G(q)-\nabla G(p_i)$. Since $G(q)=\sum_{s\in S}\beta_s g(q(s))$, the $s$-coordinate of this gradient at $\tilde q_0$ is $\beta_s g'(\tilde q_0(s))-\beta_s g'(p_i(s))$. Therefore, the active-normal relation implies
$ \sum_{i\in A}\lambda_i\big[\beta_s g'(\tilde q_0(s))-\beta_s g'(p_i(s))\big]=\zeta$ $\forall s\in S.$
Let $\Lambda:=\sum_{i\in A}\lambda_i>0$, $\mu_i:=\lambda_i/\Lambda$ for $i\in A$, and $\mu_i:=0$ for $i\notin A$. Then, $(\mu_i)_{i\in I}$ are convex weights. Dividing by $\Lambda$ gives
$\beta_s g'(\tilde q_0(s))
        =
        \frac{\zeta}{\Lambda}
        +
        \beta_s\sum_{i\in I}\mu_i g'(p_i(s))
        $ $
        \forall s\in S.$
Setting $a:=\zeta/\Lambda$ proves the desired representation.
\end{proof}

\begin{proof}[Proof of Corollary \ref{thm:multi}]
If $I$ is a singleton, choose any $p_i\in\Delta^\circ$, set $\eta_i=0$, and let $G(q)=\sum_{s\in S}q(s)\log q(s)$. Then,
$\hat B^G_{\eta_i}(p_i)=\{p_i\}$, and the multiplicative pool is $p_i$.

Suppose  $|I|\geq2$. Choose two distinct $j,k\in I$. Fix any non-uniform $\bar q\in\Delta^\circ$, and set
$$
        p_j(s):=\frac{\bar q(s)^2}{\sum_{t\in S}\bar q(t)^2},
        \quad \text{and}\quad
        p_k(s):=\frac{\bar q(s)^{-1}}{\sum_{t\in S}\bar q(t)^{-1}}
        \qquad \forall s\in S.
$$
For every $i\in I\setminus\{j,k\}$, set $p_i(s):=1/|S|$ $\forall s\in S$, so $p_i\in\Delta^\circ$. Moreover,
$\prod_{i\in I}p_i(s)
        =
        K\bar q(s)$ $\forall s\in S,$
where
$K:=
        \frac{(1/|S|)^{|I|-2}}
        {(\sum_{t\in S}\bar q(t)^2)
        (\sum_{t\in S}\bar q(t)^{-1})}$
is independent of $s$. Hence, the multiplicative pool generated by $(p_i)_{i\in I}$ is $\bar q$:
$q^\times_0(s)
        =
        \frac{\prod_{i\in I}p_i(s)}
        {\sum_{t\in S}\prod_{i\in I}p_i(t)}
        =
        \bar q(s)$ $\forall s\in S.$ Let $G(q):=\sum_{s\in S}q(s)\log q(s)$. Then, $G$ is separable and $D_G(q\Vert p_i)=\KL(q\Vert p_i)$. Define
$\eta_i:=\KL(\bar q\Vert p_i)$ $\forall i\in I .$
Since all $p_i$'s have full support, these radii are finite. By construction,
$\bar q\in\bigcap_{i\in I}\hat B^G_{\eta_i}(p_i).$

It remains to show that no other belief belongs to the intersection. The compatibility restriction used in the construction is that $\bar q$ is the weighted geometric pool of $p_j$ and $p_k$ with weights $2/3$ and $1/3$. Indeed, for every $s\in S$,
$\frac23\log p_j(s)+\frac13\log p_k(s)
        =
        \log\bar q(s)-A,$
where
$A:=
        \frac23\log\big(\sum_{t\in S}\bar q(t)^2\big)
        +
        \frac13\log\big(\sum_{t\in S}\bar q(t)^{-1}\big)$
is independent of $s$. Equivalently, on the marginal-cost scale generated by entropy,
$g'(\bar q(s))
        =
        A+\frac23 g'(p_j(s))+\frac13 g'(p_k(s))
        $ $\forall s\in S,$
where $g(x)=x\log x$. Thus, the required compatibility restriction is satisfiable by the reference beliefs above. Define
$\Phi(q):=\frac23\KL(q\Vert p_j)+\frac13\KL(q\Vert p_k).$
Using the compatibility relation,
\begin{align*}
        \Phi(q)
        &=
        \sum_{s\in S}q(s)\log q(s)
        -
        \sum_{s\in S}q(s)
        \left(\frac23\log p_j(s)+\frac13\log p_k(s)\right) =
        \sum_{s\in S}q(s)\log\frac{q(s)}{\bar q(s)}+A
        =
        \KL(q\Vert \bar q)+A .
\end{align*}

Therefore, $\Phi$ is uniquely minimized over $\Delta$ at $\bar q$, because $\KL(q\Vert\bar q)\geq0$, with equality if and only if $q=\bar q$. Now, take any
$q\in\bigcap_{i\in I}\hat B^G_{\eta_i}(p_i).$ In particular,
$\KL(q\Vert p_j)\leq\KL(\bar q\Vert p_j)$ and $
        \KL(q\Vert p_k)\leq\KL(\bar q\Vert p_k).$
Multiplying the first inequality by $2/3$, the second by $1/3$, and adding gives
$\Phi(q)\leq \Phi(\bar q).$
However, $\bar q$ is the unique minimizer of $\Phi$, so $q=\bar q$. Thus,
$\bigcap_{i\in I}\hat B^G_{\eta_i}(p_i)=\{\bar q\}.$
Since $\bar q=q^\times_0$, this proves
$\bigcap_{i\in I}\hat B^G_{\eta_i}(p_i)=\{q^\times_0\}.$
\end{proof}

\begin{proof}[Proof of Observation \ref{thm:gen}]
Let $\mathcal A$ be the closed convex hull of the collection of beliefs generated by all linear, geometric, power, and multiplicative pooling rules of the reference beliefs $(p_i)_{i\in I}$ in Sections \ref{sec:pos}--\ref{sec:power}. Since $\Delta$ is compact and finite-dimensional, $\mathcal A$ is a nonempty compact subset of $\Delta$. Fix the entropy generator
$G(q)=\sum_{s\in S} q(s)\log q(s),
$
so that $D_G(q\Vert p_i)=\KL(q\Vert p_i)$ $\forall i\in I$. Since $p_i\in\Delta^\circ$ $\forall i\in I$, the function $q\mapsto \KL(q\Vert p_i)$ is finite and continuous on $\Delta$ under the convention $0\log 0=0$. Thus, for every $i\in I$, the maximum
$M_i:=\max_{q\in\mathcal A}\KL(q\Vert p_i)
$
exists and is finite. Fix any $\varepsilon>0$ and set $\eta_i:=M_i+\varepsilon$  $\forall i\in I$. Then, $\eta_i>M_i$  $\forall i\in I$. Moreover, for every $q\in\mathcal A$ and every $i\in I$,
$\KL(q\Vert p_i)\leq M_i<\eta_i,
$
so $q\in \hat B^{\text{\normalfont\tiny KL}}_{\eta_i}(p_i)$ $\forall i\in I$. Since $q\in\mathcal A$ was arbitrary, it follows that
$\mathcal A\subseteq\bigcap_{i\in I}\hat B^{\text{\normalfont\tiny KL}}_{\eta_i}(p_i).$
\end{proof}

\subsection{Proofs from Section \ref{sec:app_finance}}

\begin{proof}[Proof of Lemma \ref{lem:finance_RS}]
Fix
$q\in\bigcap_iP_i$.

Part 1, let $(c_i\mathbf1)_i$ be a feasible full-insurance allocation. Suppose a feasible allocation $(y^i)_i$ weakly Pareto improves it and strictly improves it for at least one individual. Since $q\in P_i$,
$\E_q[u_i(y^i)]
        \geq
        V_i(y^i)
        \geq
        V_i(c_i\mathbf1)
        =
        u_i(c_i),$
with a strict inequality between the first and last terms for any individual whose maxmin utility is strictly higher. Since $u_i$ is concave, Jensen's inequality gives
$u_i(\E_q[y^i])
        \geq
        \E_q[u_i(y^i)].$
Strict monotonicity therefore implies
$\E_q[y^i]\geq c_i$ for every $i$ and a strict inequality for at least one $i$. Summing,
$\sum_i\E_q[y^i]
        >
        \sum_i c_i
        =
        \bar W.$
But feasibility implies
$\sum_i\E_q[y^i]
        =
        \E_q\left[\sum_i y^i\right]
        \leq
        \bar W,$
a contradiction.

Part 2, set
$c_i:=q\cdot\omega^i$.
No aggregate uncertainty gives
$\sum_i c_i
        =
        q\cdot\sum_i\omega^i
        =
        q\cdot(\bar W\mathbf1)
        =
        \bar W,$
so $(c_i\mathbf1)_i$ is feasible. If $y^i$ is affordable,
$q\cdot y^i\leq c_i$.
Since $q\in P_i$,
$V_i(y^i)
        \leq
        \E_q[u_i(y^i)]
        \leq
        u_i(q\cdot y^i)
        \leq
        u_i(c_i)
        =
        V_i(c_i\mathbf1).$
Thus, $c_i\mathbf1$ is optimal in individual $i$'s budget set, proving equilibrium.

Part 3, let
$(q,(c_i\mathbf1)_i)$
be a full-insurance equilibrium. Suppose
$q\notin P_i$
for some individual. Since $\omega^i\in\R^S_{++}$, $q\cdot\omega^i>0$. Hence, $c_i>0$: otherwise, any $\delta\in(0,q\cdot\omega^i)$ makes $\delta\mathbf 1$ affordable and strictly preferred to $0$. The strict hyperplane separation theorem gives $a\in\R^S$ and $b\in\R$ with
$q\cdot a
        <
        b
        \leq
        p\cdot a
        $ $\forall p\in P_i.$
Let
$d:=a-(q\cdot a)\mathbf1$.
Then,
$q\cdot d=0$
and
$\delta
        :=
        \min_{p\in P_i}p\cdot d
        >
        0.$
For sufficiently small $\varepsilon>0$, the bundle
$c_i\mathbf1+\varepsilon d$
is strictly positive and costs exactly the same as
$c_i\mathbf1$. Differentiability of $u_i$ and finiteness of $S$ imply, uniformly in $p\in P_i$,
$\E_p[
        u_i(c_i\mathbf1+\varepsilon d)]
        -
        u_i(c_i)=
        \varepsilon u_i'(c_i)p\cdot d
        +
        o(\varepsilon)\geq
        \varepsilon u_i'(c_i)\delta
        +
        o(\varepsilon)
        >
        0$
for all sufficiently small $\varepsilon$. Hence, the affordable deviation strictly raises the maxmin utility, contradicting equilibrium. Therefore,
$q\in P_i$ for every $i$.
\end{proof}

\begin{proof}[Proof of Lemma \ref{prop:finance_risk_benchmarks}]
For (i), let $\lambda_s$ be the multiplier on the resource constraint in state $s$. The FOCs are necessary and sufficient because the objective is strictly concave on the positive orthant and the constraint is linear. They give
$\mu_i\frac{p_i(s)}{x_s^i}=\lambda_s$ $\text{for every }i\in J\text{ and }s\in S.$
Thus, $x_s^i=\frac{\mu_i p_i(s)}{\lambda_s}.$
Summing over $i\in J$ and using $\sum_i x_s^i=\bar W$,
$\bar W
        =
        \frac{1}{\lambda_s}\sum_{i\in J}\mu_i p_i(s),$
so
$\lambda_s=\frac{\sum_{i\in J}\mu_i p_i(s)}{\bar W}.$
Normalizing $\lambda:=(\lambda_s)_{s\in S}$ to sum to one yields
$q^{\mathrm L}(s)
        =
        \frac{\lambda_s}{\sum_{t\in S}\lambda_t}
        =
        \frac{\sum_{i\in J}\mu_i p_i(s)}
        {\sum_{t\in S}\sum_{i\in J}\mu_i p_i(t)}
        =
        \sum_{i\in J}\mu_i p_i(s),
$
because each $p_i$ sums to one and $\sum_{i\in J}\mu_i=1$.

For (ii), the objective is strictly concave on the resource hyperplane $\sum_i x_s^i=\bar W$ $\forall s\in S$ and tends to $-\infty$ along every unbounded sequence in that hyperplane, so it has a unique finite maximizer and the FOCs are necessary and sufficient. They are
$\alpha_i p_i(s)\frac{1}{\tau_i}e^{-x_s^i/\tau_i}
        =
        \lambda_s
        $ $\forall i\in J\text{ and }s\in S.
$
Solving for $x_s^i$,
$x_s^i
        =
        \tau_i\big(\log\alpha_i+\log p_i(s)-\log\tau_i-\log\lambda_s\big).$
Summing over $i\in J$ and imposing $\sum_i x_s^i=\bar W$,
$\bar W
        =
        A+\sum_{i\in J}\tau_i\log p_i(s)
        -
        \big(\sum_{i\in J}\tau_i\big)\log\lambda_s,$
where
$A:=\sum_{i\in J}\tau_i(\log\alpha_i-\log\tau_i).$
Let $T:=\sum_{i\in J}\tau_i$. Rearranging gives
$\log\lambda_s
        =
        \frac{A-\bar W}{T}
        +
        \sum_{i\in J}\frac{\tau_i}{T}\log p_i(s).$
Since $\tau_i=\tau\mu_i$ and $\sum_i\mu_i=1$, $T=\tau$ and $\tau_i/T=\mu_i$. Thus,
$\lambda_s
        =
        K\prod_{i\in J}p_i(s)^{\mu_i}$
for a constant $K>0$ independent of $s$. Normalizing $\lambda$ to sum to one yields
$q^{\mathrm H}(s)
        =
        \frac{\prod_{i\in J}p_i(s)^{\mu_i}}
        {\sum_{t\in S}\prod_{i\in J}p_i(t)^{\mu_i}}$.
\end{proof}

\begin{proof}[Proof of Proposition \ref{thm:finance_entropy_order}]
For part (i), Corollary \ref{thm:pos1} implies that the singleton intersection is a linear pool. Remove the zero weights and denote the remaining indices by $J_{\mathrm L}$. Lemma \ref{lem:finance_RS} implies that the singleton is the unique normalized Arrow price of every full-insurance maxmin equilibrium, while Lemma \ref{prop:finance_risk_benchmarks}.(i) gives the same price in the log-utility benchmark.

For part (ii), Proposition \ref{thm:geo1} implies that the singleton intersection is a geometric pool. Remove zero weights and denote the remaining indices by $J_{\mathrm H}$. Lemma \ref{lem:finance_RS} makes this belief the unique normalized Arrow price in every full-insurance maxmin equilibrium, and Lemma \ref{prop:finance_risk_benchmarks}.(ii) gives the same price in the exponential-utility benchmark.
\end{proof}

\begin{proof}[Proof of Observation \ref{prop:finance_power_crra}]
Lemma \ref{lem:finance_RS} gives the unique full-insurance maxmin price $q_r^{\mathrm P}$. It remains to derive the risk economy benchmark. Since
$u_r'(z)=z^{-1/r}$,
the first-order condition is
$\mu_i^{1/r}p_i(s)
        (x_s^i)^{-1/r}
        =
        \lambda_s.$
Thus,
$x_s^i
        =
        \mu_i p_i(s)^r
        \lambda_s^{-r}.$
Market clearing gives
$\bar W
        =
        \lambda_s^{-r}
        \sum_i\mu_i p_i(s)^r,$
so
$\lambda_s
        =
        \bar W^{-1/r}
        \big(
        \sum_i\mu_i p_i(s)^r
        \big)^{1/r}.$
Normalizing $(\lambda_s)_s$ yields $q_r^{\mathrm P}$.
\end{proof}

\subsection{Proofs from Section \ref{sec:app_infoagg}}

\begin{proof}[Proof of Lemma \ref{prop:conditional_KL}]
Fix $i$ and suppress the index. Let
$a:=p(E)$,
$p^E:=p(\cdot|E)$,
and, when $a<1$,
$p^{E^c}:=p(\cdot|E^c)$.
Take any $q$ with $q(E)>0$ and write
$\alpha:=q(E),$ $
        r:=q(\cdot|E),$ $
        h:=q(\cdot|E^c)$
whenever $\alpha<1$. By KL chain rule:
$\KL(q\lVert p)
        =
        d(\alpha\lVert a)
        +
        \alpha\KL(r\lVert p^E)
        +
        (1-\alpha)
        \KL(h\lVert p^{E^c}),$
where
$d(\alpha\lVert a)
        :=
        \alpha\log\frac{\alpha}{a}
        +(1-\alpha)
        \log\frac{1-\alpha}{1-a}.$
For fixed $\alpha$ and $r$, the final term is minimized at
$h=p^{E^c}$. Hence,
$r\in P^E(\eta)$ if and only if there exists $\alpha\in(0,1]$ such that
$d(\alpha\lVert a)
        +
        \alpha D
        \leq
        \eta$, where $
        D:=\KL(r\lVert p^E).$

Suppose $0<a<1$. Define
$\phi(\alpha):=d(\alpha\lVert a)+\alpha D$.
Then,
$\phi'(\alpha)
        =
        \log\frac{\alpha}{a}
        -
        \log\frac{1-\alpha}{1-a}
        +
        D,$
and
$\phi''(\alpha)=1/\alpha+1/(1-\alpha)>0$.
The unique minimizer satisfies
$\frac{\alpha}{1-\alpha}
        =
        \frac{a}{1-a}e^{-D},$
so
$\alpha^*
        =
        \frac{ae^{-D}}
        {1-a+ae^{-D}}.$
Substitution gives
$\min_\alpha\phi(\alpha)
        =
        -\log(1-a+ae^{-D}).$
Therefore, $r\in P^E(\eta)$ if and only if
$1-a+ae^{-D}
        \geq
        e^{-\eta}.$
If
$\eta<-\log(1-a)$,
then
$e^{-\eta}>1-a$,
and the preceding inequality is equivalent to
$D
        \leq
        -\log\frac{e^{-\eta}-1+a}{a}.$
If instead
$\eta\geq-\log(1-a)$,
then
$e^{-\eta}\leq1-a$,
so the inequality holds for every finite $D$. Since $p^E$ has full support on the finite set $E$, every $r\in\Delta(E)$ has finite $\KL(r\lVert p^E)$. Thus,
$P^E(\eta)=\Delta(E)$.

If $a=1$, full support implies $E=S$, and the KL chain rule reduces to
$\KL(q\lVert p)
        =
        \KL(q(\cdot|E)\lVert p^E).$
The final formula then follows.
\end{proof}

\begin{proof}[Proof of Corollary \ref{cor:infoagg_regular}]
Suppose some
$q\in P_i(\eta_i)$
had $q(E)=0$. The binary component of the KL chain rule gives
$\KL(q\lVert p_i)
        \geq
        d(0\lVert p_i(E))
        =
        -\log(1-p_i(E))
        >
        \eta_i,$
a contradiction.
\end{proof}

\begin{proof}[Proof of Lemma \ref{prop:absorbing_announcement}]
Suppose first that $v\in M_i(E;\eta_i)$. Then, there exists $r\in P_i^E(\eta_i)$ such that $\E_r[Z]=v$. Repeating $v$ gives incremental payoff zero in every state, so maxmin payoff zero. For any $y\neq v$, strict propriety for the expectation gives
$\E_r[\mathsf s(y,Z)-\mathsf s(v,Z)]<0$.
Since $r\in P_i^E(\eta_i)$, the worst-case expected payoff from announcing $y$ is no greater than this strictly negative number. Hence, every $y\neq v$ is strictly worse than repeating $v$.

Conversely, suppose repeating $v$ is the unique maxmin-optimal announcement. Lemma \ref{prop:conditional_KL} implies that $P_i^E(\eta_i)$ is nonempty, compact, and convex. \citet[][Lemma 1]{agg24}, applied to this conditional belief set, implies that the unique maxmin-optimal announcement is the expectation of $Z$ under some belief in $P_i^E(\eta_i)$. Thus, there exists $r\in P_i^E(\eta_i)$ such that $\E_r[Z]=v$, which is precisely $v\in M_i(E;\eta_i)$.
\end{proof}

\begin{proof}[Proof of Proposition \ref{prop:critical_capacity}]
If $v$ lies in the convex hull of $Z(E)$, the constraint set
$\{r\in\Delta(E):\E_r[Z]=v\}$
is nonempty and compact. Since $p_i^E$ has full support, relative entropy is continuous on $\Delta(E)$, so the minimum defining $I_i(E,v)$ is attained and finite.

First suppose $0<a_i<1$. If
$\eta_i<-\log(1-a_i)$,
Lemma \ref{prop:conditional_KL} gives
$v\in M_i(E;\eta_i)
        \Longleftrightarrow
        I_i(E,v)
        \leq
        \rho_i(E,\eta_i).$
Using \eqref{eq:rho}, this is equivalent to
$e^{-I_i(E,v)}
        \geq
        \frac{e^{-\eta_i}-1+a_i}{a_i},$
or
$e^{-\eta_i}
        \leq
        1-a_i+a_i e^{-I_i(E,v)}.$
Taking minus logarithms yields
$\eta_i
        \geq
        -\log
        \left(
        1-a_i+a_i e^{-I_i(E,v)}
        \right)
        =
        \eta_i^\star(E,v).$

Because $I_i(E,v)<+\infty$,
$1-a_i+a_i e^{-I_i(E,v)}
        >
        1-a_i,$
so
$\eta_i^\star(E,v)<-\log(1-a_i)$.
Therefore, if
$\eta_i\geq-\log(1-a_i)$,
Lemma \ref{prop:conditional_KL} gives
$P_i^E(\eta_i)=\Delta(E)$ and hence
$v\in M_i(E;\eta_i)$,
while the inequality
$\eta_i\geq\eta_i^\star(E,v)$
also automatically holds. Combining the two regions proves the equivalence. If $a_i=1$, Lemma \ref{prop:conditional_KL} gives
$v\in M_i(E;\eta_i)
        \Longleftrightarrow
        I_i(E,v)\leq\eta_i,$
which is the desired formula with
$\eta_i^\star(E,v)=I_i(E,v)$.
\end{proof}

\begin{proof}[Proof of Observation \ref{prop:price_tilt}]
Suppress $i$ and $E$ from the notation. Define
$\psi(\lambda)
        :=
        \log
        \sum_s p(s)e^{\lambda Z(s)}.$
Let
$r_\lambda(s)
        :=
        p(s)e^{\lambda Z(s)-\psi(\lambda)}.$
Differentiation gives
$\psi'(\lambda)
        =
        \E_{r_\lambda}[Z]$ and $\psi''(\lambda)
        =
        \operatorname{Var}_{r_\lambda}(Z).$
Because $Z$ is nonconstant and $r_\lambda$ has full support,
$\psi''(\lambda)>0$.
Hence, $\psi'$ is strictly increasing. As
$\lambda\to-\infty$,
$\psi'(\lambda)\to\min_sZ(s)$,
while as
$\lambda\to+\infty$,
$\psi'(\lambda)\to\max_sZ(s)$.
Therefore, every interior $v$ is attained at a unique
$\lambda(v)$. Fix any $r$ satisfying $\E_r[Z]=v$. Since
$\log\frac{r_\lambda(s)}{p(s)}
        =
        \lambda Z(s)-\psi(\lambda),$
$\KL(r\lVert p)=
        \sum_s r(s)
        \log\frac{r(s)}{r_\lambda(s)}
        +
        \sum_s r(s)
        \log\frac{r_\lambda(s)}{p(s)}=
        \KL(r\lVert r_\lambda)
        +
        \lambda v
        -
        \psi(\lambda).$
Thus, $\KL(r\lVert p)
        \geq
        \lambda v-\psi(\lambda)$
 $\forall\lambda$, with equality if and only if
$r=r_\lambda$.
Taking $\lambda=\lambda(v)$ makes
$\E_{r_\lambda}[Z]=v$, so $r_\lambda$ is feasible and attains equality. Hence, $I(v)
        =
        \lambda(v)v-\psi(\lambda(v)).$
Because the preceding lower bound holds for every $\lambda$, this value is also
$\sup_\lambda\{\lambda v-\psi(\lambda)\}$.
Uniqueness follows from the strict positivity of relative entropy unless the distributions coincide.
\end{proof}

\begin{proof}[Proof of Corollary \ref{cor:rare_fragility}]
For $I>0$,
$\frac{\partial T(a,I)}{\partial a}
        =
        \frac{1-e^{-I}}
        {1-a+ae^{-I}}
        >
        0.$
Thus, the threshold is strictly increasing in the prior probability $a$ of the conditioning event. Moreover,
$T(0,I)=0$,
so 
$\lim_{a\downarrow0}
        \frac{T(a,I)}a
        =
        \left.
        \frac{\partial T(a,I)}{\partial a}
        \right|_{a=0}
        =
        1-e^{-I}$.
The first-order expansion follows immediately.
\end{proof}

\begin{proof}[Proof of Corollary \ref{cor:arrow_pricecost}]
Write
$\pi:=p_i^E(A)$.
Any $r$ satisfying
$\E_r[Z]=v$
must satisfy
$r(A)=v$.
For $v\in(0,1)$, the KL chain rule over the partition yields $\{A,E\setminus A\}$:
$\KL(r\lVert p_i^E)
        =
        d(v\lVert\pi)
        +
        v
        \KL(r(\cdot|A)\lVert p_i^E(\cdot|A))
        +(1-v)
        \KL(r(\cdot|E\setminus A)
        \lVert
        p_i^E(\cdot|E\setminus A)).$
The last two terms are nonnegative, so
$\KL(r\lVert p_i^E)\geq d(v\lVert\pi)$.
Equality is attained by preserving the reference conditional distributions within $A$ and within $E\setminus A$ while assigning total masses $v$ and $1-v$. Hence,
$I_i(E,v)=d(v\lVert\pi)$. For $v=0$, every feasible $r$ is supported on $E\setminus A$, and $D_{\rm KL}(r\Vert p^E) = -\log(1-\pi) + D_{\rm KL}(r\Vert p^E(\cdot\mid E\setminus A)). $
Thus, $I(E,0)=-\log(1-\pi)=d(0\Vert\pi).$ Similarly, $I(E,1)=-\log\pi=d(1\Vert\pi),$ so the stated formula is valid on all of $[0,1]$.
Substitution into Proposition \ref{prop:critical_capacity} yields \eqref{eq:arrowthreshold}.
\end{proof}

\begin{proof}[Proof of Theorem \ref{prop:no_learning_frontier}]
By definition,
$v\in\mathcal A(E;\eta)$
if and only if
$v\in
        M_i(C_i(s;E);\eta_i)$ $
        \forall i,\forall s\in E.$
Proposition \ref{prop:critical_capacity} makes each of these conditions equivalent to
$\eta_i
        \geq
        \eta_i^\star(C_i(s;E),v).$
For a fixed trader $i$, requiring this for every $s\in E$ is equivalent to
$\eta_i
        \geq
        \max_{s\in E}
        \eta_i^\star(C_i(s;E),v)
        =
        \bar\eta_i(E,v).$
This proves \eqref{eq:nliff}. Now, suppose the inequalities hold and the true state is $s^\star\in E$. Whenever trader $i$ is called, her information cell is
$C_i(s^\star;E)$.
Since
$v\in\mathcal A(E;\eta)$,
we have
$v\in M_i(C_i(s^\star;E);\eta_i)$.
Lemma \ref{prop:absorbing_announcement} therefore implies that repeating $v$ is her unique maxmin-optimal announcement. Because the same inherited announcement is repeated after every possible private signal compatible with $E$, the announcement reveals no additional information. Hence, the public event remains $E$. By induction, the same argument applies at every subsequent trading time and every announcement remains $v$. Since
$v\neq Z(s^\star)$,
the announcement sequence does not converge to the intrinsic value.
\end{proof}

\subsection{Proofs from Section \ref{sec:app_mechanism}}

\begin{proof}[Proof of Lemma \ref{prop:md_wedge}]
Write $z_f(s):=u(f(s))$. For every $q\in P$ and $f\in\mathcal F$, $q\cdot z_f\geq V_P(\delta_f)$. Multiplying by $\lambda(f)$ and summing gives $q\cdot\sum_f\lambda(f)z_f\geq b_P(\lambda)$. Taking the minimum over $q$ yields $V_P(\lambda)\geq b_P(\lambda)$, hence $\Gamma_P(\lambda)\geq0$. If some $q^*$ belongs to every $W_P(f)$ with $f\in\supp(\lambda)$, then $q^*\cdot z_f=V_P(\delta_f)$ for every support act. Evaluating the reduced lottery at $q^*$ gives $V_P(\lambda)\leq b_P(\lambda)$, so equality holds. Conversely, suppose $\Gamma_P(\lambda)=0$. Compactness of $P$ gives a minimizer $q^*$ of the reduced lottery. Then, $0
=
\sum_{f\in\supp(\lambda)}
\lambda(f)
\big(q^*\cdot z_f-V_P(\delta_f)\big).$
Every summand is nonnegative and every support weight is strictly positive. Hence, $q^*\cdot z_f=V_P(\delta_f)$ for every $f\in\supp(\lambda)$, so $q^*\in\bigcap_{f\in\supp(\lambda)}W_P(f)$.
\end{proof}

\begin{proof}[Proof of Lemma \ref{lem:md_entropy_tilt}]
Fix a state $s_0$. If $q(s_0)=0$, the relative-entropy chain rule for the partition $\{\{s_0\},S\setminus\{s_0\}\}$ gives $\KL(q\lVert p)\geq-\log(1-p(s_0))$. Hence, $\eta<\eta^\partial(p)$ implies that every belief in $P$ has full support.

Fix an act $f$ for which $u\circ f$ is nonconstant and write $z(s):=u(f(s))$. A least-favorable belief minimizes the linear function $q\mapsto q\cdot z$ over the compact entropy ball, so a minimizer exists. It is unique. Indeed, if two distinct beliefs were minimizers, their midpoint would also minimize the objective and, by strict convexity of relative entropy on $\Delta^\circ$, would satisfy the entropy constraint strictly. Since $z$ is nonconstant, there exist states $s,t$ with $z(s)<z(t)$. A sufficiently small transfer of probability from $t$ to $s$ would remain feasible and strictly lower the objective, a contradiction. The same perturbation argument shows that the entropy constraint binds at the unique minimizer $q_f^*$. Since $p$ is strictly feasible, Slater's condition holds. The KKT conditions are therefore necessary and sufficient. Let $\tau_f\geq0$ be the multiplier on the entropy constraint. The first-order conditions imply that, for some constant $\nu_f$, $z(s)+\tau_f[\log(q_f^*(s)/p(s))+1]+\nu_f=0$ for every state. Since $z$ is nonconstant, $\tau_f>0$. Rearranging and imposing that probabilities sum to one gives \eqref{eq:md_tilt}.

It remains to show uniqueness of $\tau_f$. If two positive multipliers $\tau$ and $\tau'$ generated the same $q_f^*$, then for every pair of states $s,t$,
$-\frac{z(s)-z(t)}{\tau}
=
\log\frac{q_f^*(s)/p(s)}{q_f^*(t)/p(t)}
=
-\frac{z(s)-z(t)}{\tau'}.$
For some pair $s,t$, nonconstancy gives $z(s)\neq z(t)$, and therefore $\tau=\tau'$.
\end{proof}

\begin{proof}[Proof of Corollary \ref{cor:md_entropy_wedge}]
Suppose $\Gamma_P(\lambda)=0$. Lemma \ref{prop:md_wedge} gives a common least-favorable belief for all support acts, and Lemma \ref{lem:md_entropy_tilt} makes that belief unique for each act. For two support acts $f,g$, let $\tau_f,\tau_g>0$ be their entropy multipliers. Taking logarithms in \eqref{eq:md_tilt} for their common least-favorable belief and eliminating the common log-likelihood ratio gives $u(g(\cdot))=(\tau_g/\tau_f)u(f(\cdot))+c\mathbf1$ for some $c\in\R$. Since $\tau_g/\tau_f>0$, their utility profiles are positive affine transformations. Applying this argument relative to one fixed support act proves the claim for the entire support. Conversely, suppose all support profiles are positive affine transformations of a fixed nonconstant profile $z$. If $z_f=a_fz+c_f\mathbf1$ with $a_f>0$, then $q\cdot z_f=a_fq\cdot z+c_f$ for every $q\in P$. Thus, all support acts have the same least-favorable beliefs. Lemma \ref{prop:md_wedge} then gives $\Gamma_P(\lambda)=0$.
\end{proof}

\begin{proof}[Proof of Proposition \ref{prop:md_exact}]
Since $0<\eta<\log2$, we have $w_\eta\in(0,1)$ and $q_\eta>0$. Multiplying $V_E$ by $2$ gives the representation $\lambda_1+\lambda_2$. Multiplying $V_A^\eta$ by $1/w_\eta$ gives $(q_\eta/w_\eta)(\lambda_1+\lambda_2)+\min\{\lambda_1,\lambda_2\}$. Thus, with $\alpha_\eta=q_\eta/w_\eta>0$, the two preferences coincide, up to positive rescaling, with the two types in \citet[][Example 2]{strack25}.

\citet[][Lemma C.13]{strack25} therefore gives $\mathrm{IMP}_\eta=\mathrm B_E\cap\mathrm B_A=\mathrm{IC}_E\cap\mathrm B_A$. Since type $E$ has expected-utility preferences, $b_E(\lambda)=V_E(\lambda)$, so membership in $\mathrm{IC}_E$ is exactly \eqref{eq:md_EIC}. For type $A$, the deterministic values are $V_A^\eta(\delta_{f^0})=0$ and $V_A^\eta(\delta_{f^1})=V_A^\eta(\delta_{f^2})=q_\eta$. Hence, $b_A(\lambda)=q_\eta(\lambda_1+\lambda_2)$. By \eqref{eq:md_binary_value}, $V_A^\eta(\lambda)\geq b_A(\lambda)$ for every $\lambda$. Consequently, the condition defining $\mathrm B_A$ for $(\lambda^E,\lambda^A)$ reduces exactly to \eqref{eq:md_AIC}. The two displayed inequalities are therefore necessary and sufficient for implementation.
\end{proof}

\begin{proof}[Proof of Corollary \ref{cor:md_capacity}]
Let $h(w):=\frac{1+w}{2}\log(1+w)+\frac{1-w}{2}\log(1-w)$, so that $h(w_\eta)=\eta$. For $w\in(0,1)$, $h'(w)=\frac12\log((1+w)/(1-w))>0$. Hence, $w_\eta$, and therefore $\Delta_\eta(\lambda^E)$ whenever $\lambda_1^E\lambda_2^E>0$, is strictly increasing in $\eta$. A Taylor expansion at zero gives $h(w)=w^2/2+w^4/12+O(w^6)$. Hence, $w_\eta=O(\sqrt{\eta})$ and $w_\eta^2=2\eta+O(\eta^2)$. Therefore
$w_\eta-\sqrt{2\eta}=(w_\eta^2-2\eta)/(w_\eta+\sqrt{2\eta})=O(\eta^{3/2})$,
which, substituted into \eqref{eq:md_rent_reduction}, proves \eqref{eq:md_small_eta}.

Finally, $h(0)=0$ and $\lim_{w\uparrow1}h(w)=\log2$, so $w_\eta\downarrow0$ as $\eta\downarrow0$ and $w_\eta\uparrow1$ as $\eta\uparrow\log2$. For fixed $t=\lambda_1^E+\lambda_2^E$, $\min\{\lambda_1^E,\lambda_2^E\}\leq t/2$, with equality if and only if $\lambda_1^E=\lambda_2^E=t/2$. The remaining claims follow from \eqref{eq:md_rent_reduction}.
\end{proof}

\begin{singlespace}
\bibliography{ref}
\bibliographystyle{apalike}
\end{singlespace}

\clearpage
\renewcommand{\thepage}{OA-\arabic{page}} \setcounter{page}{1}
{\noindent\LARGE\bf For Online Publication}

\section{Appendix: Uniqueness of Unit Cost}\label{app:unitcost_unique}

This appendix shows that the unit-cost normalization in Definition \ref{def:unitcost} is not arbitrary. Observation \ref{prop:degeneracy} shows that imposing an expected-cost budget on a posterior is vacuous: whenever a posterior can be produced, the underlying experiment can be diluted so that the same posterior is reached with arbitrarily small probability and arbitrarily small expected cost.

The constant-marginal-cost property in Lemma \ref{thm:dilution} identifies the relevant normalization. If posterior $q$ is required to occur with probability $\varepsilon$, then a posterior-level budget $\eta$ corresponds to an ex-ante expenditure $\varepsilon\eta$. Thus, the relevant object is the least expected cost of producing $q$ per unit probability of reaching it.

Fix $p\in\Delta^\circ$ and a target posterior $q$ for which $\kappa(q|p)<+\infty$. For $\varepsilon\in(0,1]$, define
$$c_q^\star(\varepsilon)
        :=
        \inf\big\{
        C(\sigma,p;Y):
        \exists y^\star
        \text{ with }
        q_{y^\star}(\sigma,p)=q,\pi_{y^\star}(\sigma,p)\geq\varepsilon
        \big\}.$$
Define the \emph{critical posterior cost}
$\Gamma(q|p)
        :=
        \lim_{\varepsilon\downarrow0}
        \frac{c_q^\star(\varepsilon)}{\varepsilon},$
whenever the limit exists.

\begin{lemma}\label{lem:critical_report_cost}
The critical posterior cost exists and satisfies
$\Gamma(q|p)
        =
        \kappa(q|p).$
\end{lemma}

A proposed posterior-cost index $\rho:\Delta\times\Delta\rightarrow\R_+\cup\{+\infty\}$ is \emph{cost-calibrated} if, for every budget $\eta\geq0$, whenever
$\rho(q|p)<\eta$,
we have
$c_q^\star(\varepsilon)\leq\varepsilon\eta$
for every sufficiently small $\varepsilon>0$, while whenever
$\eta<\rho(q|p)$,
we have
$c_q^\star(\varepsilon)>\varepsilon\eta$
for every sufficiently small $\varepsilon>0$.

\begin{obs}\label{prop:unitcost_unique}
Every cost-calibrated posterior-cost index satisfies
$\rho(q|p)
        =
        \kappa(q|p).$
Thus, the only nondegenerate cost-calibrated posterior restriction is
$\kappa(q|p)\leq\eta$.
\end{obs}

Cost calibration captures the requirement that posterior costs use the same units as the primitive expected information cost. Lemma \ref{thm:dilution} says that requiring a posterior with probability $\varepsilon$ scales the relevant ex-ante expenditure exactly by $\varepsilon$. Dividing by $\varepsilon$ is therefore the unique normalization compatible with the primitive technology.
\subsection{Proofs}
\begin{proof}[Proof of Lemma \ref{lem:critical_report_cost}]
Observation \ref{prop:unitprice}.1 gives
$c_q^\star(\varepsilon)
        \geq
        \varepsilon\kappa(q|p)$ $\forall\varepsilon,$
so
$\liminf_{\varepsilon\downarrow0}
        \frac{c_q^\star(\varepsilon)}{\varepsilon}
        \geq
        \kappa(q|p).$
Conversely, fix $\delta>0$. Observation \ref{prop:unitprice}.2 gives $\bar\varepsilon>0$ such that
$c_q^\star(\varepsilon)
        \leq
        \varepsilon(\kappa(q|p)+\delta)$
whenever
$\varepsilon\leq\bar\varepsilon$.
Hence, $\limsup_{\varepsilon\downarrow0}
        \frac{c_q^\star(\varepsilon)}{\varepsilon}
        \leq
        \kappa(q|p)+\delta.$
Since $\delta$ is arbitrary, the limsup is at most $\kappa(q|p)$. Thus, the limit exists and equals $\kappa(q|p)$.
\end{proof}

\begin{proof}[Proof of Observation \ref{prop:unitcost_unique}]
By Lemma \ref{lem:critical_report_cost},
$\Gamma(q|p)=\kappa(q|p)$.

Suppose first that
$\rho(q|p)>\Gamma(q|p)$.
Choose $\eta$ such that
$\Gamma(q|p)
        <
        \eta
        <
        \rho(q|p).$
Cost calibration gives
$c_q^\star(\varepsilon)>\varepsilon\eta$
for all sufficiently small $\varepsilon$. But convergence of
$c_q^\star(\varepsilon)/\varepsilon$
to $\Gamma(q|p)<\eta$
implies
$c_q^\star(\varepsilon)<\varepsilon\eta$
for all sufficiently small $\varepsilon$, a contradiction.

Suppose instead that
$\rho(q|p)<\Gamma(q|p)$.
Choose
$\rho(q|p)<\eta<\Gamma(q|p)$.
Cost calibration now gives
$c_q^\star(\varepsilon)\leq\varepsilon\eta$
for all sufficiently small $\varepsilon$, while convergence to
$\Gamma(q|p)>\eta$
implies the opposite strict inequality for all sufficiently small $\varepsilon$. Again there is a contradiction.

Thus,
$\rho(q|p)=\Gamma(q|p)=\kappa(q|p)$.
\end{proof}

\section{Appendix: Primal Bregman Balls}\label{sec:micro_primal}

Section \ref{sec:strict_ri} shows that uniform posterior separability is associated with the dual cost
$D_G(q\lVert p)$. The opposite/primal order has a natural scoring-rule interpretation. A scoring rule assigns a numerical reward to a reported distribution after the state is realized. Let $\iota_s$ be the Dirac vector on state $s$. Given a convex differentiable generator $G$, define
$s_G(r,s)
        :=
        G(r)
        +
        \langle\nabla G(r),\iota_s-r\rangle.$
For a data-generating distribution $q$, convexity implies that the expected score is maximized by truthful reporting. Indeed,
$\E_q[s_G(q,s)-s_G(r,s)]
        =
        D_G(q\lVert r).$
Thus, a primal Bregman divergence is the regret from using report $r$ when $q$ is the benchmark distribution \citep[e.g.,][]{savage71,banerjee05,score18,patton2020comparing}. Now, suppose $p$ is the benchmark model and $q_y$ is the posterior model proposed after signal $y$. Evaluating the regret of the benchmark against the posterior reverses the order of arguments.

\begin{definition}[Reverse uniform posterior separability]\label{def:upsprime}\normalfont
There exists a convex differentiable generator $G$ such that
$\widetilde C(\sigma,p;Y)
        =
        \sum_{y\in Y}
        \pi_y(\sigma,p)
        D_G(p\lVert q_y(\sigma,p))$
for every finite signal structure for which the expression is well defined.
\end{definition}

The reverse technology is also disciplined by Blackwell informativeness.

\begin{obs}[Blackwell monotonicity]\label{prop:rev_blackwell}
Fix $p\in\Delta^\circ$ and suppose
$q\mapsto D_G(p\lVert q)$
is convex on $\Delta^\circ$. Let $(Y,\sigma)$ induce full-support posteriors whenever the signal has positive probability. If $(Z,\tau)$ is a garbling of $(Y,\sigma)$, then
$\widetilde C(\tau,p;Z)
        \leq
        \widetilde C(\sigma,p;Y).$
\end{obs}

Thus, more informative experiments generate posterior reports that are weakly more costly to validate against the benchmark model. Define the unit cost $\widetilde\kappa(q|p)$ exactly as in Definition \ref{def:unitcost}, replacing $C$ by $\widetilde C$.

\begin{theorem}\label{thm:primal_strict_ri}
Fix $p\in\Delta^\circ$. Then, under the regularity conditions in Appendix \ref{app:conditions},
$\widetilde\kappa(\cdot|p)
        =
        D_G(p\lVert \cdot).$
Consequently, a constraint preference $(u,p,P)$ is strictly rationally inattentive under the reverse technology $\widetilde C$ and capacity $\eta$ if and only if
$P
        =
        B^G_\eta(p).$
\end{theorem}

This appendix shows that primal and dual balls encode distinct local information costs. Standard uniform posterior separability prices the movement from the reference belief to the posterior and yields $D_G(q\lVert p)$; scoring-rule validation regret prices the reference model against the posterior and yields $D_G(p\lVert q)$. Section \ref{sec:bregman} shows that reversing this local information cost can change the social pooling rule even though the aggregation principle is unchanged.
\subsection{Proofs}
\begin{proof}[Proof of Observation \ref{prop:rev_blackwell}]
Let $K(z|y)$ be the garbling kernel and define
$\tau_s(z)
        :=
        \sum_{y\in Y}
        \sigma_s(y)K(z|y).$
Write
$Y^+:=\{y:\pi_y(\sigma,p)>0\}$,
$\pi_y:=\pi_y(\sigma,p)$,
$q_y:=q_y(\sigma,p)$,
and
$\rho_z:=\pi_z(\tau,p)$.
Then, $\rho_z
        =
        \sum_{y\in Y^+}
        \pi_yK(z|y).$
For any $z$ with $\rho_z>0$, Bayes' rule gives
$$
        q_z(\tau,p)
        =
        \sum_{y\in Y^+}
        \lambda_{y|z}q_y,
        \qquad
        \lambda_{y|z}
        :=
        \frac{\pi_yK(z|y)}{\rho_z}.
$$
Thus, each garbled posterior is a convex combination of the original posteriors. Convexity of
$F_p(q):=D_G(p\lVert q)$
implies
$F_p(q_z(\tau,p))
        \leq
        \sum_{y\in Y^+}
        \lambda_{y|z}F_p(q_y).$
Multiplying by $\rho_z$ and summing over $z$,
$\widetilde C(\tau,p;Z)\leq
        \sum_z
        \sum_{y\in Y^+}
        \pi_yK(z|y)F_p(q_y)=
        \sum_{y\in Y^+}
        \pi_yF_p(q_y)
        =
        \widetilde C(\sigma,p;Y).$
\end{proof}

\begin{lemma}\label{lem:rev_small}
Let $G$ be differentiable at $p$ and at each $p+h_k$. If $h_k\to0$,
$D_G(p\lVert p+h_k)$
is well defined, and
$\nabla G(p+h_k)\to\nabla G(p)$,
then
$D_G(p\lVert p+h_k)
        =
        o(\|h_k\|).$
\end{lemma}

\begin{proof}[Proof of Lemma \ref{lem:rev_small}]
Set $q_k:=p+h_k$. Differentiability gives
$G(q_k)
        =
        G(p)
        +
        \langle\nabla G(p),h_k\rangle
        +
        o(\|h_k\|).$
Thus,
$D_G(p\lVert q_k)=
        G(p)-G(q_k)
        +
        \langle\nabla G(q_k),h_k\rangle=
        \langle
        \nabla G(q_k)-\nabla G(p),
        h_k
        \rangle
        +
        o(\|h_k\|).$
Continuity of the gradient makes the first term $o(\|h_k\|)$.
\end{proof}

\begin{proof}[Proof of Theorem \ref{thm:primal_strict_ri}]
Fix $q\in\Delta$. If $D_G(p\lVert q)=+\infty$, every signal structure producing $q$ with positive probability has infinite reverse cost, so $\widetilde\kappa(q|p)=+\infty$. Hence, suppose $D_G(p\lVert q)<+\infty$. For the lower bound, take any signal structure and designated realization inducing $q$. Reverse posterior separability and nonnegativity imply
$\widetilde C(\sigma,p;Y)
        \geq
        \pi_{y^\star}(\sigma,p)
        D_G(p\lVert q).$
Dividing and taking the infimum yields
$\widetilde\kappa(q|p)
        \geq
        D_G(p\lVert q).$

For the upper bound, use the same binary experiment as in the proof of Theorem \ref{thm:strict_ri}. For sufficiently small $\varepsilon>0$, let the designated signal occur with probability $\varepsilon$ and induce $q$, while the residual posterior is
$r_\varepsilon
        =
        \frac{p-\varepsilon q}{1-\varepsilon}.$
Then,
$\frac{\widetilde C(\sigma,p;Y)}
        {\varepsilon}
        =
        D_G(p\lVert q)
        +
        \frac{1-\varepsilon}{\varepsilon}
        D_G(p\lVert r_\varepsilon).$
Since
$\|r_\varepsilon-p\|=O(\varepsilon)$,
Lemma \ref{lem:rev_small} gives
$D_G(p\lVert r_\varepsilon)=o(\varepsilon)$.
Letting $\varepsilon\downarrow0$ yields
$\widetilde\kappa(q|p)
        \leq
        D_G(p\lVert q).$
Thus, equality holds. The final equivalence follows from the definition of strict rational inattention under $\widetilde C$:
$P
        =
        \{q:\widetilde\kappa(q|p)\leq\eta\}
        =
        \{q:D_G(p\lVert q)\leq\eta\}
        =
        B^G_\eta(p).$
\end{proof}

\end{document}